\documentclass[11pt,a4paper]{article}
\usepackage[latin1]{inputenc}
\usepackage[english]{babel}
\usepackage{amsfonts}
\usepackage{amscd}
\usepackage{amsmath}
\usepackage{theorem}
\usepackage{mathrsfs}
\usepackage{graphicx}
\usepackage{fancybox,amssymb}
\usepackage{geometry}
\usepackage{indentfirst}
\usepackage{epstopdf}
\usepackage{authblk}
\usepackage{xcolor}
\usepackage{fancyhdr}
\usepackage{url}
\usepackage[numbers]{natbib}
\usepackage[draft]{changes}
\definechangesauthor[name={Katia},color=purple]{KC}
\definechangesauthor[name={Benedetta},color=red]{BS}
\definechangesauthor[name={Julia},color=blue]{JE}
\newcommand{\mF}{\mathcal{F}}
\newcommand{\tV}{{\rm{VaR}}}
\newcommand{\R}{\mathbb{R}}
\newcommand{\N}{\mathbb{N}}
\newcommand{\mP}{\mathbf{P}}

\newcommand{\mE}{\mathbb{E}}

\newcommand{\md}{\;{\rm d}}
\newcommand{\s}{\sum\limits}

\newcommand{\w}{\wedge}

\newcommand{\bq}{\begin{eqnarray*}}
\newcommand{\eq}{\end{eqnarray*}}
\newcommand{\one}{1\mkern-5mu{\hbox{\rm I}}}

\theoremstyle{plain}
\theorembodyfont{}
\newtheorem{definition}{Definition}[section]
\newtheorem{remark}[definition]{Remark}
\newtheorem{lemma}[definition]{Lemma}
\newtheorem{proposition}[definition]{Proposition}
\newtheorem{corollary}[definition]{Corollary}

\newtheorem{example}[definition]{Example}

\makeatletter
\newenvironment{proof}{\noindent{\textit{Proof:}}}{%
\unskip\nobreak\hfil\penalty50\hskip1em\null\nobreak
$\Box$
\parfillskip=\z@\finalhyphendemerits=0\endgraf\bigskip}

\let\oldendBsp\endBsp
\def\endBsp{\unskip\nobreak\hfil\penalty50\hskip1em\null\nobreak\hfil%
$\blacksquare$\parfillskip=\z@\finalhyphendemerits=0\endgraf\oldendBsp}
\let\oldendBem\endBem
\def\endBem{\unskip\nobreak\hfil\penalty50\hskip1em\null\nobreak\hfil%
$\blacksquare$\parfillskip=\z@\finalhyphendemerits=0\endgraf\oldendBem}
\makeatother
\title{Some Optimisation Problems in Insurance with a Terminal Distribution Constraint}
\author[1]{Katia Colaneri}
\author[2]{Julia Eisenberg}
\author[3]{Benedetta Salterini}
\date{}
\affil[1]{University of Rome - Tor Vergata, katia.colaneri@uniroma2.it}
\affil[2]{TU Wien, julia.eisenberg@tuwien.ac.at}
\affil[3]{University of Florence, benedetta.salterini@unifi.it}
\normalsize
\begin{document}
\maketitle
\abstract
In this paper, we study two optimisation settings for an insurance company, under the constraint that the terminal surplus at a deterministic and finite time $T$ follows a normal distribution with a given mean and a given variance. In both cases, the surplus of the insurance company is assumed to follow a Brownian motion with drift.
\\First, we allow the insurance company to pay dividends and seek to maximise the expected discounted dividend payments or to minimise the ruin probability under the terminal distribution constraint. Here, we find explicit expressions for the optimal strategies in both cases: in discrete and continuous time settings.
\\Second, we let the insurance company buy a reinsurance
contract for a pool of insured or a branch of business. To achieve a certain level of sustainability (i.e.\ the collected premia
should be sufficient to buy reinsurance and to pay the occurring
claims)  the initial capital is set to be zero. We only allow for
piecewise constant reinsurance strategies producing a normally
distributed terminal surplus, whose mean and variance lead to a given
Value at Risk or Expected Shortfall at some confidence level $\alpha$.
We investigate the question which admissible reinsurance strategy
produces a smaller ruin probability, if the ruin-checks are due at
discrete deterministic points in time.
\bigskip
\\ {\bf Keywords}: insurance, reinsurance, dividends, optimal control, distributional constraints, value at risk, expected shortfall.
\medskip
\\{\bf 2020 MSC}: 91G05, 91B05, 93B03

\section{Introduction}
This paper investigates the problem of dividend maximisation and the problem of ruin minimisation for an insurance company who aims to achieve a certain surplus distribution at a particualr future date. Knowing the surplus distribution, for instance, at regulatory check-times can be important for the calculation of the necessary capital reserves. Measuring the solvency of a collective of risks remains one of the important tasks in insurance mathematics.
Over the years, several risk measures have been proposed and investigated concerning their properties, by adding risk constraints like, for example,  value at risk.

One of the most popular risk measures is the value of expected discounted dividends. Here, one searches for the ``optimal'' dividend strategy, i.e.\ a strategy maximising the value of expected discounted dividends up to the time when the surplus becomes negative. By considering the optimal strategy, the focus is deliberately placed on the surplus' evolution characteristics rather than on the company's managerial skills. Some results on dividend maximisation problems can be found, for instance, in \citet{astak,Shreve}. We further refer to \citet{AlbThReview, avanzi, hipp2020} and references therein for an overview of the existing results.

The optimal dividend payout strategy in the most ``unconstrained'' settings turns out to be of a barrier or of a band type, meaning that the strategy can change from ``paying the maximal possible amount'' to ``paying nothing'' in dependence on the current surplus value. This setup cannot be considered realistic or doable for an insurance company. Moreover, solvency requirements imposed by regulators may not allow to pay dividends according to the optimal, possibly bang-bang, strategy. To make the models more realistic, one needs to impose restrictions. \citet{paulsen2003optimal} studies the optimal dividend problem with a no-bankruptcy constraint -- dividends will not be paid if the surplus is below a certain barrier. An extended setting with transaction costs is analysed in \citet{bai2012optimal}.
\citet{hipp}, considers optimal dividend payment strategies under the constraint that the ruin probability stays under a given boundary. \citet{thonhauser2011} maximise the total discounted utility of dividend payments under strictly positive transaction costs.

The setting considered in the first part of this paper is novel in the following way. The surplus of an insurance company in a finite time interval is modelled by a diffusion process.
We concentrate on the dividend payments -- described by dividend rates -- with two different objective functions: expected discounted dividend payments and ruin probability. In the first case, one faces a maximisation problem, whereas in the second case the ruin probability should be minimised. The surplus can only be controlled at discrete equidistant time points. We introduce a constraint on the set of admissible strategies by requiring that the ex-dividend terminal wealth should be normally distributed with fixed exogenously given mean and variance. To the best of our knowledge, such a constraint has not been considered in an insurance optimisation problems before.
We prove that the optimal strategy in both cases should be deterministic, i.e.\ is decided at time zero. As it is intuitively clear, the strategy leading to the maximal discounted dividend value starts with high payments in the very beginning and decreases approaching the time horizon; the strategy minimising the ruin probability behaves in an opposite way.
\\The results obtained in this first part of the paper heavily rely on the very nature of dividend payments. The control is acting solely on the drift, letting the volatility unchanged. This allows to compare different strategies by comparing their paths. However, choosing a control problem with an impact on the volatility of the surplus process will not allow to use the path-comparison method and will require different, more complex techniques.
A well-established, well-investigated and still quite popular risk measure is the ruin probability -- the probability that a company, a strain of business or a pool of insured risks goes bankrupt in finite time, i.e.\ writes red numbers -- the aggregate claims exceed the collected premia. The surplus, with continuous paths or having jumps, can be controlled, for instance,  by a reinsurance, dividend payments, possible surplus investments into a dependent or independent markets.
\\A technical ruin, when the surplus becomes negative or touches zero, does not compulsory mean that the company has to entirely stop operating.
The time and the severity of ruin are completely neglected by looking solely at the ruin probability. Also, due to Solvency II requirements companies have enough reserves to bridge a certain period of unfavorable business development. For these reasons, the ruin probability might not be a desired risk measure to assess a company's performance. However, paired with some additional constraints it can help choosing a strategy which is, for instance, more risk averse in the eyes of the insurance company.\medskip
\\
In the second part of this manuscript, we are looking at the surplus of an insurance company who buys proportional reinsurance contracts of a specific type.
To control the risk exposure and to be able to meet regulatory requirements, insurance companies need to pay attention to various constraints.
For instance, \citet{bernard2009optimal, lo2017neyman, lo2017unifying, huang2019unifying} search for the optimal reinsurance strategy under a constraint (strictly positive surplus or a fixed risk measure under some prespecified boundary) on the loss at the terminal time. Optimal investment and reinsurance have been considered with constraints on the budget, see \citet{bi2014dynamic}, or on Value at Risk, see, e.g.\ \citet{choulli2001, bi2019optimal, wang2020robust}.  The problem of choosing a reinsurance strategy to minimise the ruin probability
with a Value at Risk (or a Conditional Value at Risk) constraint is considered for instance in \citet{zhang2016, chen2010optimal}  with a finite and infinite time horizon.

In this paper, we seek to find a proportional reinsurance strategy that minimises the ruin probability under a constraint imposed on the distribution of the terminal wealth. Adding a constraint on the terminal surplus has several advantages. For instance, one will be able to calculate any risk measure acting on the terminal wealth: the Value at Risk, the Expected Shortfall or the expected terminal utility, and hence address many regulatory requirements all at once.
\\We consider a finite time interval $[0,T]$, and the direct insurer can change the deductible only twice -- in the beginning and in the middle of the interval. The target is twofold: the terminal post-reinsurance surplus at $T$ should be normally distributed with given mean and variance, and the chosen admissible strategy should lead to the smallest possible ruin probability. We show that the optimal strategy is deterministic, i.e.\ is chosen at time 0, and one is always acting in a risk averse way. That is, the insurer buys less reinsurance in the beginning, in order to let the drift push the surplus upwards, and buys more reinsurance in the second half, reducing the risk of ruin shortly before the regulator's check. We briefly discuss the case where the insurer can update the reinsurance strategy three times, which provides some intuition on how to deal with more than two updates.
To the best of our knowledge, the presented approach is new in many aspects. The discrete nature of the problem and the structure of the optimal strategies, makes this setting easily applicable from a practical point of view.

The paper is organised as follows. In Section \ref{sec:max_dividend}, we introduce and solve the dividend maximisation problem. We address the ruin minimising problem in Section \ref{div:ruin}. The reinsurance optimisation problem is discussed in Section \ref{sec:reinsurance}. We conclude in Section \ref{sec:conclusion}.

\section{Maximising Dividends Under a Terminal Distribution Constraint}\label{sec:max_dividend}
In this section, we consider an insurance company who is allowed to pay dividends. The dividend rate has to be chosen in such a way that the surplus at some future deterministic time $T$ achieves a given distribution. At the same time, the value of expected discounted dividends should be maximised.

We consider a probability space $(\Omega, \mathcal F, \mP)$, a finite time horizon $T>0$ and a Brownian motion $W=(W_t)_{t \in [0,T]}$. We denoted by $\mathbb{F}$ the natural complete and right continuous filtration of $W$, and set $\mF_T=\mF$.
The surplus of the insurance company in the interval $[0,T]$ is modelled by a Brownian motion with drift as
\[
X^\mathbf{0}_t=x+\bar\mu t+\bar\sigma W_t\;,\quad t \in [0,T]
\]
where $x\ge 0$ represents the initial capital and $\bar\mu,\bar\sigma>0$. \\
The company is allowed to pay dividends in form of dividend rates $0\le c\le \xi$ for some given $\xi>0$. It means that the post-dividend process under a dividend strategy $\mathbf{c}=(c_s)_{s\in[0,T]}$ is given by
\begin{equation}\label{eq:surplus}
X^{\mathbf{c}}_t=x+\bar\mu t-\int_0^t c_s \,\mathrm{d}s +\bar\sigma W_t\;, \quad  t\in [0,T].
\end{equation}
Our objective is to determine the strategies that maximise the  expected discounted dividends and simultaneously lead to a normally distributed post-dividend terminal surplus $X_T^\mathbf{c}$. We assume that the target distribution is Gaussian with the mean $x+MT$, and the variance $\delta^2T$, for some $M\in \mathbb{R}$ and $\delta>0$.

At first, the company is only allowed to update a dividend strategy at $n\in\N$ equidistant time points $Tk/n$, $k\in\{0,..,n-1\}$ in the period $[0,T]$. An admissible strategy is a sequence $\mathbf{c}=(c_0, \dots, c_{n-1})$ of dividend rates such that for all $k=0, 1,\dots, n-1$, $c_k\in[0,\xi]$ is an $\mathcal{F}_{\frac{kT}{n}}$-measurable random variable and the total surplus at time $T$ satisfies $X_T^\mathbf{c}\sim N(x+MT, \delta^2T)$. We denote the set of admissible strategies by $\mathcal A_{(n)}$, where the subscript $(n)$ indicates the number of the allowed change points. The accumulated dividends up to time $t$ are then given by
\[
\s_{k=0}^{n-1} c_k\Big(\frac{T(k+1)}n\w t-\frac{Tk}n\w t\Big)
\]
It is worth mentioning that differently than in the classical dividend problems, see for instance \cite{astak}, dividends can be paid (up to time $T$) even if the surplus is negative.
This feature of our model alleviates, to some extent, the drawback of models stopping at the ruin time. A technical ruin does not mean that the company stops operating. In reality, some insurance companies proceed with dividend payments even during protracted crisis times. A famous example provides Munich Re, known for not reducing its dividends since at least 2006, see \cite{munich}.

The following lemma indicates the range of achievable target expectations $x+MT$ by a post-dividend Brownian surplus see equation \eqref{eq:surplus}, at time $T$.

\begin{lemma}\label{lemma:rangeM}
The parameter $M$ in the target distribution of the surplus at time $T$ has to fulfil $\bar\mu-\xi\le M\le\bar\mu$.
\end{lemma}
\begin{proof}
For any admissible dividend strategy $\mathbf{c}=(c_0, \dots, c_{n-1})\in \mathcal{A}_{(n)}$,  the distribution of the surplus in equation \eqref{eq:surplus} at time $T$ is Gaussian with mean $$x+\left(\bar\mu -\sum_{k=0}^{n-1}\frac{\mE[c_k]}{n}\right)T=x+MT.$$ Using the fact that $0\le c_k\le \xi$, for every $k\in \{0, \dots n-1\}$ we get that
$$x+\left(\bar\mu -\xi\right)T\le x+\left(\bar\mu -\frac{\sum_{k=0}^{n-1} \mE[c_k]}{n}\right)T\le x+\bar\mu T\;,$$
which proves the statement.
\end{proof}

Note that, for large values of $\xi$, the range of achievable means may include negative values. Although this is mathematically feasible, an insurance company would not pursue a strategy to achieve a negative expected surplus, but it would rather choose $M\in [0,\bar \mu]$, so to obtain a expected net profit at time $T$, even if small. Next we better identify the characteristics of admissible strategies.

\begin{proposition}\label{prop:dividend_strategies}
The set of admissible strategies $\mathcal{A}_{(n)}$ consists of $\mathbf{c}=(c_0,...,c_{n-1})$ such that $\s_{i=1}^{n-1} c_i$ is $\mF_0$-measurable, i.e.\ deterministic.
\end{proposition}
\begin{proof}
Let $\mathbf{c}=(c_0,\dots, c_{n-1})$ be an arbitrary admissible dividend strategy. The corresponding surplus at time $T$ is then given by
\begin{equation}\label{eq:surplus_1}
X_T^\mathbf{c}=x+\bar\mu T-c_0 \frac{T}{n} - \frac{T}{n}\sum_{k=1}^{n-1} c_k +\bar\sigma W_T\;.
\end{equation}

We now identify the set of dividend strategies that allow to achieve a normal distribution with mean $x+MT$ and variance $\delta^2T$.
Let $Y$ be a generic random variable with $Y\sim N(x+MT, \delta^2 T)$. Then,  for $\zeta \in \mathbb{R}$ it holds that $\mE[e^{\zeta Y}]=e^{\zeta (x+MT)+\frac{\delta^2}{2}\zeta^2T}$. Now we consider the surplus at time $T$, $X^{\mathbf c}_T$.  From \eqref{eq:surplus_1} and the fact that $\mathbf{c}$ is an admissible strategy we get that $X^{\mathbf c}_T\sim N(x+MT, \delta^2 T)$ and it holds that
\begin{align}
e^{\zeta (x+MT)+\frac{\delta^2}{2}\zeta^2T}=\mE[e^{\zeta X^C_T}]=e^{\zeta(x+\bar\mu T-c_0 T/n)}\mE[e^{\zeta\bar\sigma W_T - \zeta \sum_{k=1}^{n-1} c_k T/n}]\;.\label{eq:exp_2}
\end{align}
Let $\mathbf{Q}$ be a probability measure on $(\Omega, \mathcal{F}_T)$ equivalent to $\mP$, with the Radon-Nikodym derivative $\left.\frac{\md \mP}{\md \mathbf{Q}}\right|_{\mF_T}=e^{-\zeta\bar\sigma W_{T}+\frac{\zeta^2\bar\sigma^2T}{2}}$. Then, applying change of measure techniques in \eqref{eq:exp_2} we obtain
\begin{equation*}
\mE[e^{\zeta\bar\sigma W_{T}-\zeta \sum_{k=1}^{n-1}c_k T/n}]=e^{\frac{\zeta^2\bar\sigma^2T}{2}}\mE_\mathbf{Q}[e^{-\zeta \sum_{k=1}^{n-1}c_k T/n}]\;,
\end{equation*}
Together with \eqref{eq:exp_2}, one gets for all $\zeta\in\R$
\begin{equation*}
e^{\zeta (x+MT)+\frac{\delta^2}2\zeta^2T}=e^{\zeta(x+\bar\mu T-c_0 T/n)+\frac{\bar\sigma^2 T \zeta^2}{2}}\mE_\mathbf{Q}[e^{-\zeta \sum_{k=1}^{n-1}c_k T/n}],
\end{equation*}
leading to
\begin{equation*}
\mE_\mathbf{Q}[e^{-\zeta \sum_{k=1}^{n-1}c_k T/n}]=e^{\zeta (M-\bar\mu +c_0/n)T+\frac{\delta^2-\bar\sigma^2}2\zeta^2T}\;.
\end{equation*}
If $\delta^2-\bar\sigma^2> 0$, then by uniqueness of the moment generating functions the variable $\sum_{k=1}^{n-1}c_k \frac{T}{n}$ is normally distributed with mean $(M-\bar\mu +c_0/n)T$ and variance $(\delta^2-\bar\sigma^2)T$. Hence it has positive $\mathbf{Q}$-probability to attain negative values, which contradicts the equivalence of $\mathbf{Q}$ and $\mP$, since $\sum_{k=1}^{n-1}c_k T/n\ge 0$ $\mP$-a.s.
\\
If, instead, $\delta^2-\bar\sigma^2< 0$, there is no random variable with such a moment generating function.
\\Finally, if $\delta=\bar\sigma$, the variable $\sum_{k=1}^{n-1}c_k T/n$ must be a constant, i.e.\ deterministic.
\end{proof}

For the special case $n=2$ we obtain the following corollary.
\begin{corollary}\label{lemma:div1}
The set of admissible strategies $\mathcal{A}_{(2)}$ only consists of deterministic pairs $(c_0,c_1)$, i.e.\ $c_1$ is $\mF_0$-measurable.
\end{corollary}

Note that the dividend strategies act solely on the drift and do not affect the volatility. This fact allows to compare different strategies by looking at the surplus ``path by path''.
Another implication is that, in case $n=2$, the optimal dividend strategy is completely decided at time $t=0$; meaning that once the dividend rate $c_0$, to be valid in $[0, T/2]$, is decided, then $c_1$ is also uniquely determined at time $t=0$ so that the final distribution can be achieved. We will see in the reminder of the section that the optimal strategy is deterministic also for $n>2$.

Let now $r>0$ be the preference rate of the insurer. The return function corresponding to a strategy $\mathbf{c}=(c_0,...,c_{n-1}) \in \mathcal{A}_{(n)}$ is
\[
V^{\mathbf{c}}(x):= \mE_x\left[\sum_{k=0}^{n-1}  \frac{c_k}{r} e^{-r \frac{kT}{n}} \left(1-e^{-r \frac{T}{n}}\right)\right]\;.
\]
Note, that the dependence on the initial capital $x$ is in this setting purely nominal. As stressed before,  we do not stop our considerations at the time of ruin. The strategy will depend solely on the parameters of the surplus process and the target distribution.
\\The target of the insurance company is to find a strategy $\mathbf{c^*}=(c_0^*,...,c_{n-1}^*)\in \mathcal{A}_{(n)}$ leading to
\begin{equation}
V^{\mathbf{c^*}}(x)=\max\limits_{\mathbf{c}\in \mathcal{A}_{(n)}}  \mE_x\left[\sum_{k=0}^{n-1}  \frac{c_k}{r} e^{-r \frac{kT}{n}} \left(1-e^{-r \frac{T}{n}}\right)\right]. \label{problem:dividend_max}
\end{equation}

To analyse Problem \eqref{problem:dividend_max}, we start with the case of two periods, i.e.\ $n=2$.

\subsection{A 2-period model \label{2period}}
Suppose that the insurance company is allowed to update its dividend strategy only once, at time $T/2$.
Due to Corollary \ref{lemma:div1}, we get that the set of admissible dividend strategies $\mathcal A_{(2)}$ consists of all deterministic pairs $\mathbf{c}=(c_0,c_1)$ with $c_0,c_1\in[0,\xi]$, and such that
\[
\bar\mu-(c_0+c_1)/2=M.
\]
As a direct consequence of the fact that $c_0, c_1\in \mathcal{F}_0$, it must also hold that $\bar\sigma^2=\delta^2$, otherwise the target distribution would not be reachable.
In the next step, we investigate how to determine the optimal strategy.

\begin{proposition}
The optimal strategy $\mathbf{c}^*=(c_0^*, c_1^*)$ is given by
\begin{align*}
&c_0^*=\xi\w2(\bar\mu-M)
\\
&c_1^*=\begin{cases}
0 &\mbox{ if $2\bar\mu-2M\le \xi$}\\
2\bar\mu-2M-\xi &\mbox{ if $2\bar\mu-2M> \xi$}
\end{cases}\;.
\end{align*}
\end{proposition}
\begin{proof}
We consider the problem
\begin{equation}\label{eq:opt_problem}
\max_{(c_0, c_1)\in \mathcal{A}_{(2)}} \frac{c_0 T}{r} \left(1-e^{-r T/2}\right)+\frac{c_1 T}{r} e^{-r T/2}\left(1-e^{-r T/2}\right)\;.
\end{equation}
It is easy to see that, for $r>0$, the discounting coefficient in the first period, $1-e^{-r T/2}$, is larger than in the second period, $e^{-r T/2}(1-e^{-r T/2})$. Therefore, to maximise the discounted dividends, $c_0$ must be chosen as big as possible. Taking into account that  $0\leq c_0\leq \xi$ and that $c_0+c_1=2(\bar \mu - M)$,
we get that $c_0=\min \left(\xi, 2(\bar\mu-M)\right) $, and consequently, $c_1=2(\bar\mu-M)-\xi$ if $c_0=\xi$ and $c_1=0$ if $c_0=2(\bar\mu-M)$.
\end{proof}

To summarise the result, in a two-period setting, the optimal dividend strategy pays dividends at the maximum rate in the first period, and then adjusts the strategy to achieve the target distribution in the second period. Such behaviour is justified by the effect of discounting which has a larger impact in the time interval $[T/2, T]$.

\subsection{An $n$-period model}
We now extend our analysis to an $n$-period framework. That is, the dividend strategy can be adjusted $n$ times in the interval $[0,T]$. Recall that, according to Proposition \ref{prop:dividend_strategies},  strategies are not necessarily deterministic, but the sum of dividend rates is.

To better explain the mechanism for the computation of the optimal dividend strategy, we consider an example with $n=3$.
\begin{example}\label{example:3period_dividends}
Let $n=3$ and let $\mathbf{c}=(c_0, c_1, c_2)$ be an admissible strategy. The expected discounted total dividends are given by
\[
\frac{c_0}{r}(1-e^{-rT/3})
+\frac{e^{-rT/3}(1-e^{-rT/3})}r\mE\Big[c_1+c_2e^{-rT/3}\Big]\;.
\]
We easily see that, that due to discounting ($r>0$), the strategy $c_0$ to be applied in the first period has a larger weight than the others, hence, as in the two period model,  it would be optimal to choose it the largest possible.
Taking into account that $x+MT=x+\bar\mu T-\frac{(c_0 +\mE[c_1+c_2])T}{3}$, and that $c_k\in [0, \xi]$ for $k=0, \dots, 2$,
we have that
\begin{align*}
&c_0=\begin{cases}
3(\bar\mu-M) &\mbox{ if $3(\bar\mu-M)\le \xi$}\\
\xi &\mbox{ if $3(\bar\mu-M)> \xi$}
\end{cases}\;,
\end{align*}
equivalently, $c_0=\min (3\bar\mu-3M, \xi)$. Now we move to the choice of $c_1, c_2$. After choosing $c_0$ we get that $\mE[c_1+c_2]=c_1+c_2=\max(0, 3(\bar\mu-M)-\xi)$, according to Proposition \ref{prop:dividend_strategies}. If $c_0=3\bar\mu-3M$, since $c_1$ and $c_2$ are nonnegative, it holds that  $c_1=c_2=0$. If instead, $c_0=\xi$, using the same argument like for $c_0$, we choose $c_1$ and $c_2$ so that $c_1$ is the largest possible value according to the constraints, i.e. $c_1=\min(3(\bar\mu-M)-\xi, \xi)$, and $c_2=\max(3(\bar\mu-M)-2\xi, 0)$. Put in other words, if $2\xi\le 3\bar\mu-3M< 3\xi$, then $c_0=c_1=\xi$ and $c_2=3\bar\mu-3M-2\xi$. If $\xi<3\bar\mu-3M< 2\xi$, at time $T/3$ we determine both $c_1$ and $c_2$, depending on the current surplus so that
$$\bar\mu T-(c_1+c_2)T/3= MT+\xi T/3.$$
We stress that because $c_1+c_2$ must be deterministic, we immediately get that $c_2$ is $\mF_{T/3}$ measurable. That means, once $c_1$ is found, then $c_2$ is also determined, so that the constraint on the distribution is satisfied. Moreover, the value $3(\bar\mu-M)-\xi$ is the biggest possible choice for $c_1$.
\\
The deterministic strategy $\mathbf{c^*}=(c_0,c_1,c_2)$ where
\begin{equation}\label{eq:dividends_3period}
\begin{cases}c_0^*=\min(\xi, 3 (\bar\mu - M))&\\
c_1^*= \max\big(\min (\xi, 3 (\bar\mu - M)-\xi ),0\big)&\\
c_2^*= \max (3 (\bar\mu - M)-2\xi, 0)&
\end{cases}
\end{equation}
fulfils all necessary conditions.

Next, we show that we cannot find a different, possibly stochastic, strategy with a higher expected discounted dividends value, meaning that the optimal strategy is indeed deterministic.

\smallskip

Let $\mathbf{c^*}=(c_0^*, c_1^*, c_2^*)$ be the strategy in \eqref{eq:dividends_3period} and let $\tilde{\mathbf{c}}=(\tilde c_0,\tilde c_1,\tilde c_2)\in\mathcal{A}_{(3)}$ be an arbitrary admissible strategy, i.e.\ such that $\tilde c_m\in [0,\xi]$ for $m=0,1,2$, and $X^{\tilde{\mathbf{c}}}_T\sim N(x+MT, \delta^2 T)$.
\\Then, there exist two random variables $d_1,d_2$ such that $\mE[d_1+d_2]=c_0^*-\tilde c_0\ge 0$, because $c_0^*$ is the largest possible dividend rate,  and $\tilde c_1=c_1^*+d_1$, $\tilde c_2=c_2^*+d_2$. It holds that
\begin{align*}
V^{\tilde{\mathbf{c}}}(x)&=\tilde c_0(1-e^{-\frac{r T}3})+(1-e^{-\frac{r T}3})e^{-\frac{r T}3}\mE[\tilde c_1 +\tilde c_2 e^{-\frac{r T}3}]
\\&=V^{\mathbf{c}^*}(x)-(c_0^*-\tilde c_0)(1-e^{-\frac{r T}3})+(1-e^{-\frac{r T}3})e^{-\frac{r T}3}\mE[d_1 +d_2 e^{-\frac{r T}3}]
\\&=V^{\mathbf{c}^*}(x)-(1-e^{-\frac{r T}3}) \left(\mE[d_1+d_2] - e^{-\frac{r T}3}\mE[d_1 +d_2 e^{-\frac{r T}3}]\right)
\end{align*}
where in the last equality we have used the fact that $(c_0^*-\tilde c_0)=\mE[d_1+d_2]$.
\\
If $\mE[d_2]< 0$, then, $\mE[c_2^*]>\mE[\tilde c_2]\ge 0$; hence, necessarily  $c_2^*=3(\bar \mu-M)-2\xi>0$ and $c_1^*=c_0^*=\xi$. Since $c^*_0+c^*_1+c^*_2=\tilde c_0 + \mE[\tilde c_1] + \mE[\tilde c_2]$ we get that $\tilde c_0+\mE[\tilde c_1]>c_0^*+\mE[c_1^*]=2\xi$ leading to a contradiction.  \smallskip
\\Then, it must hold that $\mE[d_2]\ge 0$. Now we have two cases:
\begin{itemize}
\item[i.] if $\mE[d_1]<-e^{-rT/3}\mE[d_2]$, then it is immediate that $\mE[d_1+d_2] - e^{-\frac{r T}3}\mE[d_1 +d_2 e^{-\frac{r T}3}]>0$ and then $V^{\tilde{\mathbf{c}}}\le V^{ \mathbf{c^*}}$;
\item[ii.] if $\mE[d_1]\ge-e^{-rT/3}\mE[d_2]$, we get that
\[
c_0^*-\tilde c_0=\mE[d_1 +d_2 ]\ge \mE[d_1 +d_2 e^{-\frac{r T}3}] \ge e^{-\frac{r T}3}\mE[d_1 +d_2 e^{-\frac{r T}3}]\;,
\]
which implies that $V^{\tilde{\mathbf{c}}}\le V^{ \mathbf{c^*}}$.
\end{itemize}
To conclude we observe that if $\mE[d_2]>0$ then the inequality is strict and the strategy $(c^*_0, c^*_1, c^*_2)$, is optimal. If $\mE[d_2]=0$ we get that either $\mE[d_1]>0$ in which case the inequality is strict again, or $\mE[d_1]=0$ which corresponds to the case where $\tilde{\mathbf{c}}=\mathbf{c}$.
 \end{example}
{}\hfill $\blacksquare$
\\The above example provides the argument for computing the optimal dividend strategy in an $n$-period framework.

\begin{proposition}
Let $n$ be the number of sub-periods in the interval $[0,T]$ and let
\begin{equation}\label{eq:k_dividends}
\kappa:=\min\{m\ge 0:\; n(\bar\mu-M)<(m+1)\xi\}\;.
\end{equation}
Then, an optimal strategy $\mathbf{c}^*=(c_0^*, c_1^*, \dots, c_{n-1}^* )$ is given by
\begin{equation}
\label{eq:dividen_strategy_n}
\begin{cases}
c_0^*=...=c_{\kappa-1}^*=\xi \;, &
\\
c_{\kappa}^*= n (\bar\mu-M)-\kappa \xi\;,&\\
c_{\kappa+1}^*=...c_{n-1}^*=0\;. &
\end{cases}
\end{equation}
\end{proposition}

\begin{proof}
Assume first $\kappa=n-1$, then obviously the optimal strategy is $(\xi,...,\xi)$.\medskip
\\Let now $\kappa<n-1$ and let $\tilde{\mathbf{c}}=(\tilde c_0,...,\tilde c_{n-1})$ be an admissible strategy.
Like in Example \ref{example:3period_dividends}, there exist $d_1,...,d_{n-1}$ such that $\tilde c_m=c_m^*+d_m$ for $m\in\{1,...,n-1\}$ and $\s_{m=1}^{n-1}\mE[d_m]=c_0^*-\tilde c_0\ge 0$. Then we have that
\[
V^{\tilde{\mathbf{c}}}(x)=V^{{\mathbf{c^*}}}(x)-(c_0^*-\tilde c_0)(1-e^{-rT/n})+(1-e^{-rT/n})\s_{m=1}^{n-1} e^{-rTm/n} \mE[d_m]\;.
\]
 Note that since $c^*_m=\xi$ for all $m\le \kappa-1$, and $c^*_m=0$ for all $m>\kappa+1$, it must hold $d_m\le 0$ for $m\le \kappa-1$ and $d_m\ge 0$ for $m\ge \kappa+1$.
\\Now we observe that, the function $t \to \s_{m=1}^{n-1} e^{r t \frac{\kappa-m}{n}}\mE[d_m]$ is decreasing, and hence it attains its maximum at $t=0$, i.e. $\s_{m=1}^{n-1} \mE[d_m]\ge\s_{m=1}^{n-1} e^{rT\frac{\kappa-m}{n}}\mE[d_m]$ . Therefore, we conclude that
\[
c_0^*-\tilde c_0=\s_{m=1}^{n-1} \mE[d_m]\ge e^{-rT\frac{\kappa}{n}}\s_{m=1}^{n-1} \mE[d_m]\ge e^{-rT\frac{\kappa}{n}}\s_{m=1}^{n-1} e^{rT\frac{\kappa-m}{n}}\mE[d_m]\;.
\]
The strict inequality holds true if there is at least one $m$ with $\mE[d_m]\neq 0$. If instead $\mE[d_m]=0$ for all $m=1, \dots, n-1$, then strategies $
\tilde{\mathbf{c}}$ and $\mathbf{c}^*$ coincide, i.e. in particular $d_m=0$ almost surely for all $m=0, \dots, n-1$. This leads to $V^{\tilde{\mathbf{c}}}< V^{{\mathbf{c^*}}}$ if $\tilde{\mathbf{c}}\not \equiv {\mathbf{c^*}}$.
\end{proof}

\begin{remark}[Continuous time]
This procedure allows to extend the setting to continuous time.
\\We denote by $\mathcal{A}_{(\infty)}$ the set of admissible strategies, consisting of the $\mathbb{F}$-adapted processes $\mathbf{c}=(c_s)_{s\in[0,T]}$ with $0\le c_s\le \xi$ and $X_T^{\mathbf{c}}$ from \eqref{eq:surplus} normally distributed with mean $x+MT$ and variance $\delta^2T$. Letting $n\to\infty$ in the $n$-period models, the optimal strategies as given in \eqref{eq:dividen_strategy_n} converge to a deterministic strategy in continuous time:
\[
c^*_s=\begin{cases}\xi&\mbox{: $t\le T\w t^*$},\\
0&\mbox{: $t> T\w t^*$},
\end{cases}
\]
where $t^*=(\bar \mu-M)T/\xi$. We assume $t^*<T$.
\\Let $\tilde{\mathbf{c}}=(\tilde c_s)_{s\in[0,T]}$ be an admissible strategy and define $d_s:=\tilde c_s-c^*_s$. Since we would like to achieve the same final distribution with strategies $\mathbf{c}^*$ and $\tilde{\mathbf{c}}$, it must hold that $\mE[\int_0^T d_s]=0$. Moreover, it is clear that $d_s\le 0$ for $s\le t^*$, $d_s\ge 0$ for $s>t^*$. As for the $n$-period models we get
\begin{align*}
V^{\tilde{\mathbf{c}}}(x)&=\mE_x\Big[\int_0^T e^{-rs}\tilde c_s\md s\Big]=V^{\mathbf{c^*}}(x)+\mE_x\Big[\int_0^T e^{-rs}d_s\md s\Big]
\\&=V^{\mathbf{c^*}}(x)+e^{-rt^*}\mE_x\Big[\int_0^T e^{-r(s-t^*)}d_s\md s\Big]
\\&\le V^{\mathbf{c^*}}(x)+e^{-rt^*}\mE_x\Big[\int_0^T d_s\md s\Big]=V^{\mathbf{c^*}}(x)\;.
\end{align*}
A strict inequality holds true if $\mE[d_s] \neq 0$ for all  $s\in \mathcal{T}$ where $\mathcal{T}\subseteq [0,T]$ is a Lebesgue measurable non-zero set.
\\Therefore, in continuous time it is optimal to pay on the maximal rate as long as possible, and to pay nothing afterwards.
{}\hfill $\blacksquare$
\end{remark}

Note, that we have only considered the case of dividend rates. However, it is also possible to allow for lump sum payments.
Then, because $r>0$ it is clear that one should pay the amount $(\bar \mu-M)T$ directly at time zero in both discrete and continuous time settings.

Considering the setting with dividend rates may be more preferable for reputational reasons.
Indeed, distributing the dividends over the whole period $[0,T]$ rather than paying a lump sum at the beginning of the period, may give a better impression to shareholders. This is guaranteed in our model by the upper bound $\xi$ on the admissible dividend rates. The value of $\xi$ is a management decision: in our setting it has to be small enough to distribute dividends over the whole period, and large enough to achieve the target distribution.

\section{Dividends Minimising the Ruin Probability \label{div:ruin}}
At first glance, the title of this sections sounds controversial. Indeed, paying dividends increases the probability of ruin, and in many settings the optimal dividend strategy even leads to a certain ruin. However, the constraint put on the terminal distribution allows to find a non-zero dividend strategy that minimises the ruin probability over the set of admissible strategies.
\\We consider again an insurance company who pays dividends and aims to achieve a target distribution of the terminal surplus at time $T$. However, we now assume that the objective of the insurer is to minimise the ruin probability.

We consider the same setting like in Section \ref{sec:max_dividend} with a surplus, after dividends, described by equation \eqref{eq:surplus}. We recall that the set of achievable target means is given by $\bar\mu-\xi<M<\bar\mu$ (see Lemma \ref{lemma:rangeM}) and that the set of admissible strategies $\mathcal{A}_{(n)}$ is the set of all strategies $\mathbf{c}=(c_0, \dots, c_{n-1})$, where $\sum_{m=1}^{n-1}c_m$ is $\mathcal{F}_0$-measurable (see Proposition \ref{prop:dividend_strategies}), and $c_0+\sum_{m=1}^{n-1}c_m=n(\bar \mu - M)$.

The goal of the insurance company is to minimise the ruin probability, which is given by

\begin{equation}
\label{problem:dividend_min}
\min \mP[\inf\limits_{0\le t\le T}X^{\mathbf{c}}_t<0]
\end{equation}
over all admissible dividend strategies $\mathbf{c}\in \mathcal{A}_{(n)}$.

Like in Section \ref{sec:max_dividend}, we begin by addressing Problem \eqref{problem:dividend_min} in a two-period framework.
\subsection{A 2-period model}
The set of admissible strategies is denoted by $\mathcal A_{(2)}$, given in Section \ref{2period}. That is all admissible strategies are of the form $\mathbf{c}=(c_0,c_1)$ with $c_0,c_1\in[0,\xi]$ deterministic (see Corollary \ref{lemma:div1}) and $c_0+c_1=2(\bar \mu -M)$.
We target to minimise the ruin probability in the time interval $[0,T]$, i.e.\
\[
p(\mathbf{c},x):=\mP\left[\inf\limits_{0\le t\le T}X^{\mathbf{c}}_t<0\right]\;,
\]
over all $\mathbf{c} \in \mathcal A_{(2)}$. Note that differently than in Section \ref{sec:max_dividend} the dependence on the initial capital $x$ is crucial in this setting.
\begin{proposition}\label{prop:dividend_ruin2}
Let $\mathbf{c}=(c_0, c_1)$ and $\tilde{\mathbf{c}}=(\tilde c_0, \tilde c_1)$ be two admissible strategies, i.e.\ $X^{\mathbf{c}}, X^{\tilde{\mathbf{c}}}\sim N(x+MT, \delta^2T)$. We assume that $c_0>\tilde c_0$. Then, $\tilde{\mathbf{c}}$ {\em is better than} $\mathbf{c}$, in the sense that $$p(\tilde{\mathbf{c}}, x)<p(\mathbf{c},x).$$
\end{proposition}
\begin{proof}
We first observe that at time $T$, both strategies $\mathbf{c}$ and $\tilde{\mathbf{c}}$ lead to the same distribution of the final surplus,
i.e.\ $X^{\mathbf{c}}_T,X^{\tilde{\mathbf{c}}}_T\sim N(x+MT, \delta^2T)$. Then, we have
\begin{align*}
\inf\limits_{0\le t\le T}X^{\mathbf{c}}_t
&=\inf\limits_{0\le s\le T}
\begin{cases}
x+(\bar\mu-c_0) s+\bar\sigma W_s &\mbox{: if $s\le \frac{T}{2}$}\\
x+(\bar\mu-c_0) \frac{T}{2}+\bar\sigma W_{s}+(\bar\mu-c_1)(s-\frac{T}{2})&\mbox{: if $s\in( \frac{T}{2},T]$}
\end{cases}
\\&=\inf\limits_{0\le s\le T} \begin{cases}
x+(\bar\mu-c_0) s+\bar\sigma W_s &\mbox{: if $s\le \frac{T}{2}$}\\
x+\bar\mu s+c_0(s-T)+\bar\sigma W_{s}-2(\bar \mu-M)(s-\frac{T}{2})&\mbox{: if $s\in( \frac{T}{2},T]$}
\end{cases}
\\&=\inf\limits_{0\le s\le T}
\begin{cases}
X^{\tilde{\mathbf{c}}}_s+(\tilde c_0-c_0)s &\mbox{: if } s\le \frac{T}{2}\\
X^{\tilde{\mathbf{c}}}_s+(c_0-\tilde c_0)(s-T) &\mbox{: if } s\in( \frac{T}{2},T]
\end{cases}
\\&< \inf\limits_{0\le t\le T}X^{\tilde{\mathbf{c}}}_t\;.
\end{align*}
Therefore, for all $x>0$ we get that $p(\mathbf{c}, x)>p(\tilde{\mathbf{c}},x)$.
\end{proof}

As a consequence of Proposition \ref{prop:dividend_ruin2}, we get that $c_0$ should be chosen as the smallest possible value.
This leads to the following result.

\begin{corollary}\label{cor:dividend_ruin2}
In a two-period framework, the ruin minimising dividend strategy is $\mathbf{c^*}=(c_0^*, c_1^*)$ where
\begin{align*}
\begin{cases}c_0^*=\max(2(\bar\mu-M)-\xi, 0)\;,&\\
c_{1}^*=\min( \xi, 2(\bar\mu-M))\;.&
\end{cases}
\end{align*}
\end{corollary}

\subsection{An n-period model}
The extension to $n$-periods is obtained by replicating the reasoning of Proposition \ref{prop:dividend_ruin2} and Corollary \ref{cor:dividend_ruin2}.

\begin{proposition}\label{prop:dividend_ruinn}
Let $k:=\min\{m\ge 0:\; n(\bar\mu-M)<(m+1)\xi\}$. Then, the ruin minimising dividend strategy $\mathbf{c}=(c_0,...,c_{n-1})$ fulfils
\begin{align*}
\begin{cases}
c_{n-1}=...=c_{n-k}=\xi\;, &
\\c_{n-k-1}=n(\bar\mu-M)-k\xi\; &
\\c_0=...=c_{n-k-2}=0\;. &
\end{cases}
\end{align*}
\end{proposition}

\begin{remark}[Continuous time]
Letting $n\to\infty$ will produce the following optimal strategy:  we define $t^*$ as the time that realises $(\bar\mu-M)T=\xi t^*$. Then, the optimal dividend rate is $c_t=0$ for all $0<t<t^*$ and $c_t=\xi$ for $t\ge t^*$.
\end{remark}

Like in the dividend maximisation problem in Section \ref{sec:max_dividend}, the ruin minimising strategy is deterministic. An intuitive consequence of the result above (Proposition \ref{prop:dividend_ruinn}) is that the optimal strategy that minimises ruin probability is also the strategy that minimises the value of expected discounted dividends.

\section{Reinsurance With a Target Terminal Distribution}\label{sec:reinsurance}
In this section, we change the setting considered in Sections \ref{sec:max_dividend} and \ref{div:ruin}. We consider an insurance company who buys reinsurance for a certain branch of their business or a pool of insured claims.

We consider a probability space $(\Omega,\mF,\mP)$, and a fixed time horizon $T>0$. Let $Z$ be a random variable representing a claim size having positive finite first and second moments denoted by $\mE[Z]=\mu$ and $\mE[Z^2]=\mu_2$, respectively. We assume that the surplus of the insurance company is described by a Brownian motion with drift, approximating a Cramer-Lundberg model like, e.g., in \citet[p.\ 226]{HS},
\[
X_t=x+\lambda \eta \mu t+ \sqrt{\lambda\mu_2} W_t\;, \quad t \in [0,T],
\]
where $\lambda,\eta>0$ and $W=(W_t)_{t \in [0,T]}$ is a Brownian motion. We also define by $(\mF_t)_{t \in [0,T]}$ natural filtration of the Brownian motion, under the usual hypotheses.
The insurance company is allowed to buy proportional reinsurance with retention $b\in [0,1]$ to mitigate the losses. We assume that the reinsurance premium is calculated via the expected value principle, that is $[(1 + \theta)\lambda \mu - \mE[r(Z, b)]$ where $r(Z,b)=(1-b)Z$. Then, the premium rate that remains to the insurer is $$c(b) = \lambda(1 + \theta)\mE[r(Z, b)] - \lambda\mu(\theta - \eta)$$ with $c(0)<0$, see, e.g.\ \cite[Ch.\ 2.2]{HS} for more details.

Under a reinsurance strategy $\mathbf{b}=(b_s)_{s\in[0,T]}$, the surplus is given by
\[
X_t^{\mathbf{b}}= x + \lambda\mu\int_0^t(\theta  b_s  - (\theta - \eta)) \md s +\sqrt{\lambda\mu_2}\int_0^t  b_s   \md W_s\;, \quad t \in [0,T].
\]
We denote by
$$\bar{X}_t^\mathbf{b}:=X^\mathbf{b}_t-x, \quad t \in [0,T],$$
the net value of collective, i.e.\ the part of the surplus that only accounts for insurance/reinsu\-rance premia and claims.

The insurance company wants to make sure that the collected premia
are sufficient (in a certain sense) to buy reinsurance, if necessary, and to pay the occurring claims. To achieve such level of {\em sustainability} the target of the insurance is to choose a reinsurance strategy such that at time $T$ the distribution of the net collective is normal with mean $MT$, for some small $M>0$ and variance $\delta^2T$. To gain some intuition on the choice of $M$ and $\delta$, we may interpret $M$ as a (small) positive target gain, and $\delta$ is fixed to fulfil $\mP[-\bar X^B_T>\ell]\le 1-\alpha$ for some given $\ell >0$ and $\alpha \in (0,1)$. The latter is a condition on the Value at Risk (VaR) at the confidence level $\alpha$ (for instance $\alpha=99.5\%$). In particular, $\ell$ represents the loss that the insurer can bear with at most probability $1-\alpha$. This can be interpreted as the required capital ensuring the system's solvency. Aiming at $\bar{X}_T^B\sim N(MT,\delta^2T)$ as a target distribution is justified, for instance,  by the existence of the closed form formulas for the VaR or Expected Shortfall (ES) for Gaussian random variables, which  can be easily calculated\footnote{Denoting by $L_T$ the terminal loss at time $T$, we immediately get that
\begin{equation*}
  \tV_{\alpha}(L_T)=-MT+\delta\sqrt{T}\Phi^{-1}(\alpha),
\end{equation*}
and
\begin{equation*}
	ES_{\alpha}(L_T)= - MT + \delta \sqrt{T} \frac{\varphi(\Phi^{-1}(\alpha))}{1-\alpha},
\end{equation*} for $\alpha \in (0,1)$, where $\varphi$ and $\Phi$ denote the density and the cumulative distribution function of the standard normal, respectively.}.

To reach the target distribution, the insurance company follows a {\em sustainable} strategy, that is reinsurance is only financed through premia, and we additionally require that reinsurance strategies do not produce a negative premium rate for the insurer. This condition is, in spirit, similar to the self-financing condition which is often assumed in finance. Indeed, a branch of business or a pool of insured is considered a closed system, where the insurance company does not intervene by injecting or withdrawing additional capital.

Our next step is to define the set of possible controls leading to the desired distribution. We let $\mathcal{B}$ denote the set of strategies $\mathbf{b}=(b_t)_{t \in [0,T]}$ with $b_t\in[0,1]$ for all $t \in [0,T]$, that are adapted to $(\mF_t)_{t \in [0,T]}$ and such that $\bar{X}^{\mathbf{b}}_T\sim N(MT,\delta^2T)$. \\
Note that, in particular, deterministic controls make the terminal distribution of the net collective surplus Gaussian, see Example \ref{ex:example1}.

\begin{example}[Deterministic controls]\label{ex:example1}
Let $\mathbf{b}=(b(t))_{t \in [0,T]}$ be a continuous deterministic reinsurance strategy, with $b(t)\in [0,1]$ for all $t \in [0,T]$\footnote{In this example we use the notation $b(t)$ in place of $b_t$ to emphasise the deterministic nature of the strategy.}. Then $\mathbf{b}$ is an admissible control if the following two conditions hold:
\begin{equation}
\label{eq:normal}
\begin{cases}
\lambda\mu\int_0^T (\theta  b(s) - (\theta - \eta))\md s=MT\;, &
\\\lambda\mu_2\int_0^T b(s)^2\md s=\delta^2T\;. &
\end{cases}
\end{equation}

\begin{figure}[t]
\begin{center}
\includegraphics[width=0.5\textwidth]{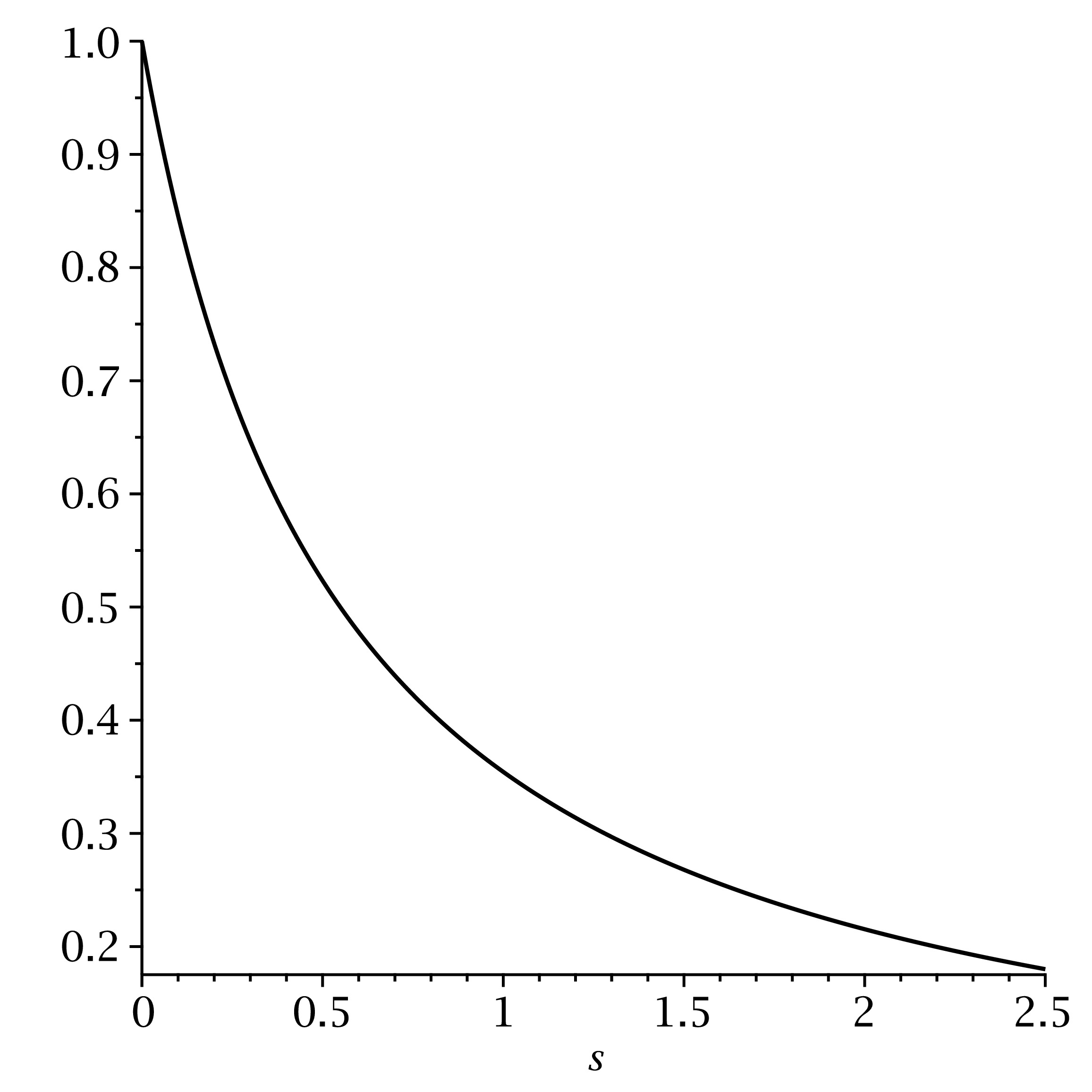}
\caption{The strategy $b(t)=\frac{A}{A+Ct}$ with $A= 3.61$, $C= 6.58$. \label{fig:1}}
\end{center}
\end{figure}
To make an example, $b(t)=\frac{A}{A+Ct}$ is an admissible control  for constants $A, C$ which satisfy
\begin{align*}
& \int_0^T \frac{A}{A+Cs} \md s =\frac{A^2}{C}\ln\Big(\frac{A+CT}{A}\Big)=\frac{M+ \lambda\mu(\theta - \eta)T}{\lambda \theta \mu}\;,
\\& \int_0^T\frac{A^2}{(A+Cs)^2}\md s=\frac A C\Big(1-\frac A{A+CT}\Big)=\frac{\delta^2}{\lambda\mu_2}\;.
\end{align*}
For the parameter set given by $\mu = 0.05$, $\lambda = 1$, $\eta = 0.3$, $\theta = 0.5$, $\mu_2 = 0.05$, $M = 0.06$, $\delta = 0.15$, $T = 2.5$,  we get that  $A= 3.613$, $C= 6.5837$. The strategy is illustrated in Figure \ref{fig:1}. \hfill $\blacksquare$
\end{example}

In the sequel, we restrict to the case where reinsurance strategies can be updated only at deterministic time points. In fact, we concentrate on the case $n=2$ and we refer to this case as the {\em two period model}. The reason is that differently than in the dividend case, reinsurance controls affect both the drift and the volatility. Therefore, in this case, pathwise comparison is not possible anymore, and the problem requires different techniques. In the case $n=2$, we are still able to obtain an explicit solution with probabilistic methods. However, the problem becomes immediately more complicated when we increase the number of periods (see Section \ref{sec:3periods}), even in case we restrict to deterministic strategies.

\subsection{Admissible strategies in a 2-period model}
We denote the set of admissible strategies by $\mathcal{B}_{(2)}$, where, like before, the subscript indicates the number of strategy updates up to time $T$.
An admissible strategy is a pair $\mathbf{b}=(b_0, b_1)$, where $b_0$ is $\mF_0$-measurable and $b_1$ is $\mF_{T/2}$ measurable; In this setting the retention level is updated only once, at time $T/2$. Hence, at time $T$ the net surplus satisfies
\begin{align*}
\bar{X}^\mathbf{b}_T&=\frac{\lambda\mu\theta T}{2}(b_0+b_1)-\lambda\mu(\theta-\eta)T+\sqrt{\lambda \mu_2}b_0 W_{T/2}+\sqrt{\lambda \mu_2}b_1(W_T-W_{T/2})
\\&=\bar{X}_{T/2}^{b_0}+b_1\frac{T}{2}\lambda\mu\theta-\lambda\mu(\theta-\eta)\frac{T}{2}+b_1\sqrt{\lambda \mu_2}(W_T-W_{T/2})\;,
\end{align*}
where $\bar{X}_{T/2}^{b_0}=\lambda\mu(\theta b_0-\theta+\eta) \frac{T}{2}+\sqrt{\lambda \mu_2}b_0 W_{T/2}$.

\medskip

The set of admissible strategies is characterised in the lemma below. In particular, we show that admissible strategies are deterministic.

\begin{lemma}
The set $\mathcal{B}_{(2)}$ consists of all strategies $\mathbf{b}=(b_0, b_1)$ where $b_0, b_1$ are both $\mathcal{F}_0$-measurable, taking values in $[0,1]$, and satisfying the following two conditions:
\begin{align}
\begin{split}\label{conditions1}
&b_1=2\frac{M+\lambda\mu (\theta - \eta)}{\lambda\theta \mu }-b_0
\\&b_1^2=\frac{2\delta^2}{\lambda\mu_2 }-b_0^2\;.
\end{split}
\end{align}
\end{lemma}

\begin{proof}
Recall that for any normally distributed random variable $Y$ with mean $MT$ and variance $\delta^2T$, the moment generating function is given by
$\mE\left[e^{\zeta Y}\right]=e^{\zeta MT+\frac{1}{2}\zeta^2\delta^2T}$ for all $\zeta\in \R$. Let $\tilde W_{{T}/{2}}=W_T-W_{{T}/{2}}$; then $W_{T/2}$ and $\tilde W_{{T}/{2}}$ are independent.
Since  $\mathbf{b}$ is chosen so that $X^{\mathbf{b}}_T\sim N(MT,\delta^2 T)$, it holds that
\[
\mE[e^{\zeta X_T^B}]= e^{\zeta MT+\frac{1}{2}\zeta^2\delta^2T}
\]
for all $\zeta\in \R$.
Now, we let $\mathbf{Q}$ be a probability measure equivalent to $\mP$, with the Radon-Nikodym derivative
\[
\left.\frac{\md \mathbf{Q}}{\md \mP}\right|_{\mathcal F_T}=  e^{r\sqrt{\lambda \mu_2}b_1 \widetilde{W}_{{T}/{2}}- \frac{r^2}{2}\lambda \mu_2 b_1^2\frac{T}{2}}.
\]
Using the independence of $\tilde W_{T/2}$ and $W_{T/2}$ and the change of measure we get that
\begin{align*}\label{eq:MGFX}
e^{\zeta MT+\frac{1}{2}\zeta^2\delta^2T}&=\mE[e^{\zeta \bar X_T^B}]= \mE\left[e^{\zeta \bar X^{b_0}_{T/2}+r\lambda\mu(\theta b_1-\theta+\eta)\frac{T}{2}+\zeta \sqrt{\lambda\mu_2}b_1 \tilde W_{T/2}}\right]
\\&=e^{\zeta\lambda\mu (b_0\theta-\theta+ \eta)T/2+\frac{\zeta^2\lambda\mu_2 b_0^2 T}{4}}\mE\left[e^{\zeta\lambda\mu(\theta b_1-\theta+\eta)\frac{T}{2}+\zeta\sqrt{\lambda\mu_2}b_1 \tilde W_{T/2}}\right]
\\&=e^{\zeta \lambda\mu (b_0\theta-\theta+ \eta)T/2+\frac{\zeta^2\lambda\mu_2 b_0^2 T}{4}}\mE_{\mathbf{Q}}\left[e^{\zeta\lambda\mu(\theta b_1-\theta+\eta)\frac{T}{2}+\frac{\zeta^2}{2}\lambda \mu_2 b_1^2\frac{T}{2}}\right]
\end{align*}
for all $\zeta\in\R$. This can be simplified to
\[
\mE_{\mathbf{Q}}[e^{\zeta\lambda\mu(\theta b_1-\theta+\eta)\frac{T}{2}+\frac{\zeta^2}{2}\lambda \mu_2 b_1^2\frac{T}{2}}]=e^{\zeta MT-\zeta\lambda\mu (b_0\theta-\theta+ \eta)\frac{T}{2}+\frac{1}{2}\zeta^2\delta^2T-\frac{\zeta^2\lambda\mu_2 b_0^2 T}{4}}\;.
\]
Deriving the above expression with respect to $\zeta$ and letting $\zeta=0$ we see that all moments of $b_1$ correspond to the moments of a normal distribution, meaning that the moment generating function of $b_1$ (written as a power series with the moments as coefficients) corresponds to that of a normal distribution. Therefore, we conclude
\begin{equation*}
b_1 \sim N\left( 2\frac{M+\lambda \mu (\theta-\eta)}{\lambda \mu \theta}-b_0,\frac{2\delta^2}{\lambda \mu_2}-  b_0^2 \right)\;.
\end{equation*}
However, this is impossible because $b_1$ can attain values only in $[0,1]$ $\mP$-a.s. (hence also $\mathbf{Q}$-a.s.), which means that $b_1$ must be constant.
\end{proof}

Because $b_0, b_1$ can take values only in $[0,1]$, it is clear that not all arbitrary values of $M$ and $\delta$ are reachable. In the next lemma we specify the ranges of $M$ and $\delta$.

\begin{lemma}\label{prop}
If there exist $b_0,b_1\in[0,1]$ such that Condition \eqref{conditions1} holds, then the target mean $M$ and the variance $\delta>0$ satisfy:m
\begin{gather}0\le  M  \le \lambda\mu \eta,\nonumber\\
\left( \frac{M+\lambda \mu(\theta-\eta)}{\lambda \mu \theta}\right)^2, \frac{M}{\lambda\mu\theta} \le  \frac{\delta^2}{\lambda \mu_2}  \le \min \left\{ 2 \left( \frac{M+\lambda \mu(\theta-\eta)}{\lambda \mu \theta}\right)^2,1 \right\}.\label{eq:bound_vol}
\end{gather}
\end{lemma}
\begin{proof}
From Conditions \eqref{conditions1}, and the fact that $b_0,b_1$ take values in $[0,1]$, we get that $0\le M\leq \lambda \mu \eta$ and that $\frac{\delta^2}{\lambda\mu_2}\leq 1$.
Using again the conditions \eqref{conditions1} and substituting the value of $b_1$ into the second equation we get that $b_0$ must solve
\[2b_0^2-4b_0\frac{M+\lambda\mu(\theta-\eta)}{\lambda\mu\theta}+4\left(\frac{M+\lambda\mu(\theta-\eta)}{\lambda\mu\theta}\right)^2-\frac{2\delta^2}{\lambda\mu}=0 \]
Imposing the existence of a real solution leads to
\[
\frac{\delta^2}{\lambda \mu_2}\ge\left( \frac{M+\lambda \mu(\theta-\eta)}{\lambda \mu \theta}\right)^2.
\]
Then, using the fact that $b_0$ must take non-negative values leads to the bound:
  \[\frac{\delta^2}{\lambda \mu_2} \le \min \left\{2 \left( \frac{M+\lambda \mu(\theta-\eta)}{\lambda \mu \theta}\right)^2,1\right\}.\]
\end{proof}

Notice that because we do not allow for arbitrage and require $\eta<\theta$, to ensure the existence of a solution at least for the case $\eta=\theta$, we must have that $\frac{\delta^2}{\lambda \mu_2}\ge\big( \frac{M}{\lambda \mu \theta}\big)^2$, which is guaranteed by \eqref{eq:bound_vol}.

There is a clear trade-off between increasing the profits and reducing the risk. This is due to the fact that a reinsurance strategy controls both the mean and the volatility. Under a reinsurance strategy the mean and the volatility move into the same direction: increasing the retention level makes the mean larger, but also the volatility. This observation has important consequences for the ruin probability. Indeed, a bigger retention level would make the drift of the net collective larger, meaning that it potentially can push the surplus away from zero; however, at the same time, it increases the riskiness by making the volatility larger. For instance, considering the parameters $\mu=0.22;  \mu_2=0.05; \eta=0.3; \theta=0.35;\lambda=2$, $M=0.08<0.132$, we get that the admissible values of $\delta$ vary in the range $[0.2094,0.2962]$. If an insurance company aims at getting an expected gain of $8\%$ at the end of the observation period, it has to account for a relatively large risk of at least $21\%$.

We can write the range for $\delta$ as
\[\left( 1+ \frac{M-\lambda \mu \eta}{\lambda \mu \theta}\right)^2 \le \frac{\delta^2}{\lambda \mu_2} \le \min \left\{ 2 \left( 1+ \frac{M- \lambda \mu\eta}{\lambda \mu \theta}\right)^2,1 \right\}.\]
From this expression  it is clear that if the target return is close to $\lambda \mu \eta$, the variance $\delta^2$ is approximately $\lambda \mu_2$, which corresponds to the case where no reinsurance is bought.

\subsection{Ruin probabilities in a 2-period model}
For the case $n=2$, the pairs of strategies that satisfy Conditions \eqref{conditions1} are of the type $(b_0, b_1)$ and $(b_1, b_0)$.

We assume that $\frac{T}{2}$ and $T$ are the regulatory authorities' inspection dates. A reinsurance strategy is chosen so that the probability of having a positive surplus at both dates is maximised.
\\We now give a definition of ruin within this setting. We say that {\em the ruin occurs if the insurance company showcases a negative surplus at any of the time points $T/2$ or $T$}.
Then, an equivalent formulation of the problem is:
\begin{center}
Find a reinsurance strategy that minimises the ruin probability.
\end{center}

In mathematical terms, the problem is formulated as follows. Let $\mathbf{b}=(b_0,b_1)$ and $\tilde{\mathbf{b}}=(b_1, b_0)$ be the two admissible strategies. Without loss of generality, we assume that $b_0\le b_1$. For each strategy we define the corresponding {\em survival probabilities}:
\begin{align*}
p(\mathbf{b})&=\mP\left[\bar{X}^{\mathbf{b}}_{T/2}>0, \bar{X}^{\mathbf{b}}_T>0\right]\;,\\
p(\tilde{\mathbf{b}})&=\mP\left[\bar{X}^{\tilde{\mathbf{b}}}_{T/2}>0, \bar{X}^{\tilde{\mathbf{b}}}_T>0\right]\;.
\end{align*}
Our objective is to decide which of these two probabilities, $p(\mathbf{b})$ or $p(\tilde{\mathbf{b}})$, is the largest.

The table below illustrates survival probability for different values of $\eta<\theta$ so that $M$ and $\delta$ are achievable for $T=1, \lambda=2, \mu=0.22, \mu_2=0.05, \theta=0.35, M=0.05, \delta=0.2$.  The last two columns suggest that $p(\mathbf{b})>p(\tilde{\mathbf{b}})$. This result is proved in Proposition \ref{prop:optimal_strategy} below.

\begin{center}
\begin{tabular}{|p{0.14\textwidth}|p{0.14\textwidth}|p{0.14\textwidth}|p{0.14\textwidth}|p{0.14\textwidth}|}
	\hline
 $\eta$          &$b_0$       & $b_1$  &  ${\bf{p}}(b_0,b_1)$ &   ${\bf{p}}(b_1,b_0)$	\\
\hline
$0.25$ & 0.4448 &0.7760 &0.4088 & 0.5117 \\
\hline
$0.26$ & 0.3339 &0.8298 &0.3772 &0.5372 \\
\hline
$0.27$ & 0.2468 &0.8597 &0.3485 &0.5561\\
\hline
$0.28$ & 0.1715 &0.8778 & 0.3154 &0.5720 \\
\hline
$0.29$ & 0.1038 &0.8884 &0.2637 &0.5857 \\
\hline
$0.3$ & 0.0416 & 0.8935 & 0.1254 &0.5967 \\
\hline

\end{tabular}
\end{center}

\begin{proposition}\label{prop:optimal_strategy}
Let $b_0<b_1$.
Then the strategy $(b_1, b_0)$ is better than the strategy $(b_0, b_1)$, i.e. $p(\tilde{\mathbf{b}})>p(\mathbf{b})$.
\end{proposition}

\begin{proof}
Let $(W_{t})_{t \ge 0}$ and $(\hat W_t)_{t \ge 0}$ be two independent Brownian motions and denote
\begin{align*}
\bar{X}^{b}_{T/2}&=\lambda \mu \left(\theta (1-b)-\eta\right)T/2+\sqrt{\lambda \mu_2} bW_{T/2}\;,\\
\hat{X}^{b}_{T/2}&=\lambda \mu \left(\theta (1-b)-\eta \right)T/2+\sqrt{\lambda \mu_2} b\hat{W}_{T/2}\;,
\end{align*}
Then, the survival probabilities can be rewritten as
\begin{align}
\label{eq:pr1}
\begin{split}
p(\mathbf{b})&=\mP\left[\bar{X}^{b_0}_{T/2}>0, \bar{X}^{b_0}_{T/2}+\hat{X}^{b_1}_{T/2}>0\right]\;,\\
p(\tilde{\mathbf{b}})&=\mP\left[\bar{X}^{b_1}_{T/2}>0, \bar{X}^{b_1}_{T/2}+\hat{X}^{b_0}_{T/2} >0\right]\;,
\end{split}
\end{align}
having set $\hat W_{T/2}=W_T-W_{T/2}$. The advantage of this representation stands in the fact that for every $b \in [0,1]$, $\bar{X}^b_{T/2}$ and  $\hat{X}^b_{T/2}$ are independent.
We observe that there exist standard Brownian motions $ W^0$ and $W^1$ such that
\begin{align}
\begin{split} \label{eq:W0}
&\sqrt{\lambda\mu_2}b_0 W_t+\sqrt{\lambda\mu_2}b_1 \hat W_t=\sqrt{\lambda\mu_2(b_0^2+b_1^2)}W_t^0\;,
\\& \sqrt{\lambda\mu_2}b_1 W_t+\sqrt{\lambda\mu_2}b_0 \hat W_t=\sqrt{\lambda\mu_2(b_0^2+b_1^2)} W_t^1\;,
\end{split}
\end{align}
for all $t \in [0,T]$. We now let
\begin{align*}
 Y^0_t&:=\lambda\mu\theta(b_0+b_1)t-2\lambda\mu(\theta-\eta)t+\sqrt{\lambda\mu_2(b_0^2+b_1^2)} W_t^0
\\
Y^1_t&:=\lambda\mu\theta(b_0+b_1)-2\lambda\mu(\theta-\eta)t+\sqrt{\lambda\mu_2(b_0^2+b_1^2)} W_t^1
\end{align*}
for all $t \in [0,T]$. Due to Equations \eqref{eq:W0} we get that for all $t \in [0,T]$,  $Y^0_t=Y^1_t=2Mt+\sqrt{2\delta^2} \hat W_t$, hence they are identically distributed.

Next, we write $\bar{X}_{T/2}^{b_0}$ and $\bar{X}_{T/2}^{b_1}$ in terms of  $Y^0_{T/2}$ and $Y^1_{T/2}$. Since $\bar{X}_{T/2}^{b_0}$ and $\bar{X}_{T/2}^{b_1}$ are normally distributed, we have that
\begin{align*}
&\bar{X}_{T/2}^{b_0}=\rho Y^0_{T/2}+Z^0\;, \quad \rho:=\frac{Cov(\bar{X}_{T/2}^{b_0}, Y^0_{T/2})}{Var(Y^0_{T/2})}=\frac{Var[\bar{X}_{T/2}^{b_0}]}{ Var[Y^0_{T/2}]}=\frac{\lambda\mu_2b_0^2}{2\delta^2 }\;,
\\&\bar{X}_{T/2}^{b_1}=\gamma Y^1_{T/2}+Z^1\;,\quad \gamma:=\frac{Cov(\bar{X}_{T/2}^{b_1}, Y^1_{T/2})}{Var[Y^1_{T/2}]}=\frac{Var[\bar{X}_{T/2}^{b_1}]}{Var[Y^1_{T/2}]}=\frac{\lambda\mu_2b_1^2}{2\delta^2 }=1-\rho\;,
\end{align*}
where $Y^0_{T/2}$ and $Z^0$, $Y^1_{T/2}$ and $Z^1$ are independent, since they are all normally distributed and $Cov(Y^0_{T/2},Z^0)=Cov(Y^1_{T/2},Z^1)=0$. \\Expectations and variances of $Z^0$ and $Z^1$ are given by
\begin{equation}
\label{ew}
\begin{cases}
\mE[Z^0]=\mE[\bar{X}_{T/2}^{b_0}-\rho Y^0_{T/2}]=\lambda\mu(\theta b_0-\theta+\eta)T/2-2\rho M T/2\;,&
\\\mE[Z^1]=\mE[\bar{X}_{T/2}^{b_1}-\gamma Y^1_{T/2}]=\lambda\mu(\theta b_1-\theta+\eta)T/2-2\gamma M T/2 =-\mE[Z^0]\;,&
\\Var[Z^0]=\lambda\mu_2b_0^2T/2-2\rho^2\delta^2 T/2=2\delta^2\rho\gamma T/2 \;, &
\\
Var[Z^1]= \lambda\mu_2 b_1^2T/2-2\gamma^2 \delta^2 T/2=2\delta^2\gamma\rho T/2\;. &
\end{cases}
\end{equation}
Using Fubini's theorem, we get
\begin{align*}
p(\mathbf{b})&=\mP\left[Y^0_{T/2}+ \frac{Z^0}{\rho}>0, \  Y^0_{T/2}>0\right]=\mP\left[ \frac{Z^0}{\rho}>-Y^0_{T/2},  \ Y^0_{T/2}>0\right]\\
&= \int_0^{\infty} \left(1-\Phi\left(\frac{-y-\frac{\mE\left[Z^0\right]}{\rho}}{\sqrt{\delta^2 \frac{\gamma}{\rho}T}}\right)\right) f_{Y^0_{T/2}}(y) \md  y\;,\\
p(\tilde{\mathbf{b}})&=\mP\left[ Y^1_{T/2}+ \frac{Z^1}{\gamma}>0,  \ Y^1_{T/2}>0\right] =\mP\left[ \frac{Z^1}{\gamma}>-Y^1_{T/2},  \ Y^1_{T/2}>0\right]\\
&= \int_0^{\infty} \left(1-\Phi\left(\frac{-y+\frac{\mE\left[Z^0\right]}{\gamma}}{\sqrt{\delta^2 \frac{\rho}{\gamma}T}}\right)\right) f_{Y^1_{T/2}}(y) \md y\;,
\end{align*}
where $\Phi$ is the standard normal distribution, $f_{Y^0_{T/2}}(y)=f_{Y^1_{T/2}}(y)$ are the densities of the random variables $Y^0_{T/2}$ and $Y^1_T/2$, respectively, and we have used that $\mE[Z^1]=-\mE[Z^0]$.
Since $\Phi$ is increasing, we consider the crucial quantities
\begin{equation}\label{z}
z_0:=\frac{-y-\frac{\mE\left[Z^0\right]}{\rho}}{\sqrt{\delta^2 \frac{\gamma}{\rho}\frac{T}{2}}} \quad \mbox{and}\quad z_1:=\frac{-y+\frac{\mE\left[Z^0\right]}{\gamma}}{\sqrt{\delta^2 \frac{\rho}{\gamma}\frac{T}{2}}}\;.
\end{equation}
We have the following two cases:
\\ $[i.]$
Assume first that $\mE[Z^0]\le 0$, with $\mE[Z^0]$ from \eqref{ew}.
\\Since $1-2\rho=\frac{\lambda \mu_2}{\delta^2}\left(\frac{\delta^2}{\lambda \mu_2}-b_0^2\right)=\frac{\lambda \mu_2}{\delta^2} \frac{b_1^2-b_0^2}{2}>0$, it holds that $z_0>z_1$ for all $y >0$. Then, we can immediately conclude, that $p(\mathbf{b})<p(\tilde{\mathbf{b}})$, and hence the strategy $\tilde{\mathbf{b}}=(b_1, b_0)$ is better than the strategy $\mathbf{b}=(b_0, b_1)$.  \medskip
\\$[ii.]$ Assume next that $\mathbb{E}[Z^0]>0$, with $\mE[Z^0]$ from \eqref{ew}.
\\Note that since $b_0\le b_1$ and $b_0,b_1\in[0,1]$, then either $b_0=b_1=1$, in which case there is nothing to prove since $\mathbf{b}$ and $\tilde{\mathbf{b}}$ are equal,  or it cannot hold that $b_0=1$. The latter implies that
there always exists an $y^*\in(0,+\infty)$ such that for $y<y^*$, it holds that $z_0<z_1$, and the opposite holds true for $y>y^*$, see, e.g.\ the right panel in Figure \ref{fig:survival_dominance}.\\
Now, we consider the functions $\eta to b_0(\eta)$ and $\eta \to b_1(\eta)$ for $\eta\in [0,\theta]$. We know that
\[
b_0(\eta)+b_1(\eta)=2\frac{M-\lambda\mu(\theta-\eta)}{\lambda\mu\theta}\quad \mbox{and}\quad b_0(\eta)^2+b_1(\eta)^2=2\frac{\delta^2}{\lambda\mu_2}
\]
meaning that $b_0'(\eta)+b_1'(\eta)=\frac{2}{\theta}>0$ and $b_0'(\eta)b_0(\eta)+b_1'(\eta)b_1(\eta)=0$. Consequently, since $b_0<b_1$ we get that  $b_1'(\eta)>0$ and $b_0'(\eta)<0$, that is to say, $b_0$ is decreasing with respect to $\eta$ and $b_1$ is increasing. 
\smallskip

We let $\eta=\theta$ and note that, in this case for any $a \in [0,1]$ it holds that 
\[
\bar X^b_t=b\big\{\lambda\mu\theta +\sqrt{\lambda\mu_2} W_t\big\}=\frac ba \bar X^a_t\;.
\]
Using this fact, equations \eqref{eq:pr1} and $b_0<b_1$, we get
\begin{align*}
p(\mathbf b)=\mP\big[\bar X_{T/2}^{b_0}>0, \bar X_{T/2}^{b_0}+\hat X_{T/2}^{b_1}\big]&= \mP\big[\bar X_{T/2}^{b_1}>0, \frac{b_0}{b_1}\bar X_{T/2}^{b_0}+\hat X_{T/2}^{b_0}\big]
\\&= \mP\big[\bar X_{T/2}^{b_1}>0, \frac{b_0^2}{b_1^2}\bar X_{T/2}^{b_1}+\hat X_{T/2}^{b_0}\big]
\\&< \mP\big[\bar X_{T/2}^{b_1}>0, \bar X_{T/2}^{b_1}+\hat X_{T/2}^{b_0}\big]
\\&=p(\tilde{\mathbf b})\;.
\end{align*}
which proves the statement in case $\eta=\theta$.

Now, we let $\mathbf{b}(\eta)$ and $\tilde{\mathbf{b}}(\eta)$ be, respectively, the strategies $(b_0, b_1)$ and $(b_1, b_0)$ corresponding to $\eta\in (0, \theta)$. Assume there is an $\bar\eta\in(0,\theta)$ such that $p(\mathbf b(\bar\eta))>p(\tilde{\mathbf b}(\bar\eta))$. Then, by the intermediate value theorem there exists an $\eta^*\in(\bar\eta,\theta)$ such that $p(\mathbf b(\eta^*))=p(\tilde{\mathbf b}(\eta^*))$.
\\Assume that $\mathbf b(\eta^*)\neq \tilde{\mathbf b}(\eta^*)$. Let $\bar X$ be a random variable, independent of $Z^1$ and $Z^0$ with $\bar X\sim N(MT,\delta^2 T)$.
\begin{align*}
0&=p(\mathbf b(\eta^*))-p(\tilde{\mathbf b}(\eta^*))
\\&=\mP\left[\gamma \bar X + Z^1>0, \bar X>0\right]-\mP\left[\eta \bar X + Z^0>0, \bar X>0\right]\\
&=\mP\left[\max\left(-\frac{Z^1}{\gamma},0\right)<\bar X < \max\left(-\frac{Z^0}{\eta},0\right)\right]>0\;.
\end{align*}
The last inequality follows from the fact that $\bar X$, $Z^1$ and $Z^0$ are normally distributed. Hence, this contradiction yields that $\mathbf b(\eta^*)=\tilde{\mathbf b}(\eta^*)$. However, since it holds that $b_1'>0$ and $b_0'<0$ for $b_0<b_1$, that means $b_0(\bar\eta)>b_0(\eta^*)$ and $b_1(\bar\eta)<b_1(\eta^*)$, contradicting $\mathbf b(\eta^*)=\tilde{\mathbf b}(\eta^*)$.\\ That concludes the fact that $\tilde{\mathbf{b}}=(b_0, b_1)$ is always better than $\mathbf{b}$.
\end{proof}

In the following, we discuss the situations where $\mE[Z^0]<0$ and derive sufficient conditions for this to hold.

\begin{lemma}
If \begin{equation}\label{eq:condition1}
\frac{\mu \eta}{\mu_2}\le \frac{M}{\delta^2} \ \mbox{ and } \ \frac{\mu \theta}{\mu_2}\le \frac{2M}{\delta^2}
\end{equation}
then $\mE[Z^0]\le 0$ with $\mE[Z^0]$ given in \eqref{ew}.
\end{lemma}

\begin{proof}
We observe that $\mE[Z_0]\le 0$ is equivalent to
\begin{align}\label{eq:condition2}
- \frac{\lambda \mu_2 b_0^2}{\delta^2} M T +\lambda \mu (\theta b_0 - \theta + \eta)T \le 0,
\end{align}
for all $b_0 \in [0,1]$, where we have substituted  $\rho=\frac{\lambda \mu_2 b_0^2}{2\delta^2}$. To show that \eqref{eq:condition2} holds for all $b_0\in [0,1]$, we consider the function
\[
F(b)=- \frac{ \mu_2 b^2}{\delta^2} M  + \mu (\theta b - \theta + \eta).
\]
This function is concave and has a maximum at $b^*=\frac{\mu\theta\delta^2}{2\mu_2M}>0$.
We observe that $F(0)<0$ and that $F(1)=- \frac{ \mu_2}{\delta^2} M  + \mu \eta$ which is negative if the first of condition \eqref{eq:condition1} holds true. Moreover, under the second condition in \eqref{eq:condition1} we also get that  $b^*>1$, which guarantees that $\mE[Z]<0$.
\end{proof}

Condition \eqref{eq:condition1} is meaningful in terms of insurance and reinsurance premia. Indeed, it tells us that the reinsurance is cheap and the income from the direct insurance premia is high. In this case the result of Proposition \ref{prop:optimal_strategy} is intuitively clear: choosing a bigger retention level (less reinsurance) in the first period is better in terms of survival probability (i.e.\ minimises the ruin at times $T/2$ and $T$).

To better understand the different cases (i.e.\ $\mE[Z_0]\le 0$ and $\mE[Z_0]>0$) we let
\begin{align*}
G_0(y)= \left(1-\phi\left(\frac{-y-\frac{\mE\left[Z^0\right]}{\rho}}{\sqrt{\delta^2 \frac{1-\rho}{\rho}T}}\right)\right) f_{Y^0_{T/2}}(y),\\
G_1(y)= \left(1-\phi\left(\frac{-y+\frac{\mE\left[Z^0\right]}{\gamma}}{\sqrt{\delta^2 \frac{1-\gamma}{\gamma}T}}\right)\right) f_{Y^1_{T/2}}(y)
\end{align*}
for all $y>0$, so that
\begin{align*}
p(\mathbf{b})=\int_0^{\infty} G_0(y) \mathrm{d} y, \qquad p(\tilde{\mathbf{b}})=\int_0^{\infty} G_1(y) \mathrm{d} y
\end{align*}

Figure \ref{fig:survival_dominance} represents the densities of the survival probability (i.e. $G_0(y)$ and $G_1(y)$) relative to the strategy $\mathbf{b}=(b_0, b_1)$ (dashed line) and the strategy $\tilde{\mathbf{b}}(b_1, b_0)$ (solid line), under given parameters.

\begin{figure}[t]
  \centering
\includegraphics[width=.48\textwidth]{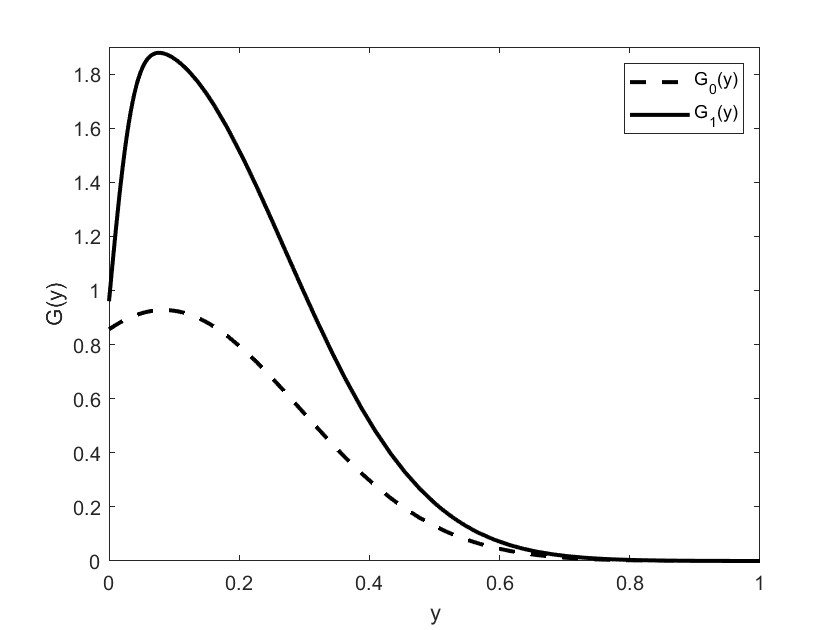}
  \includegraphics[width=.48\textwidth]{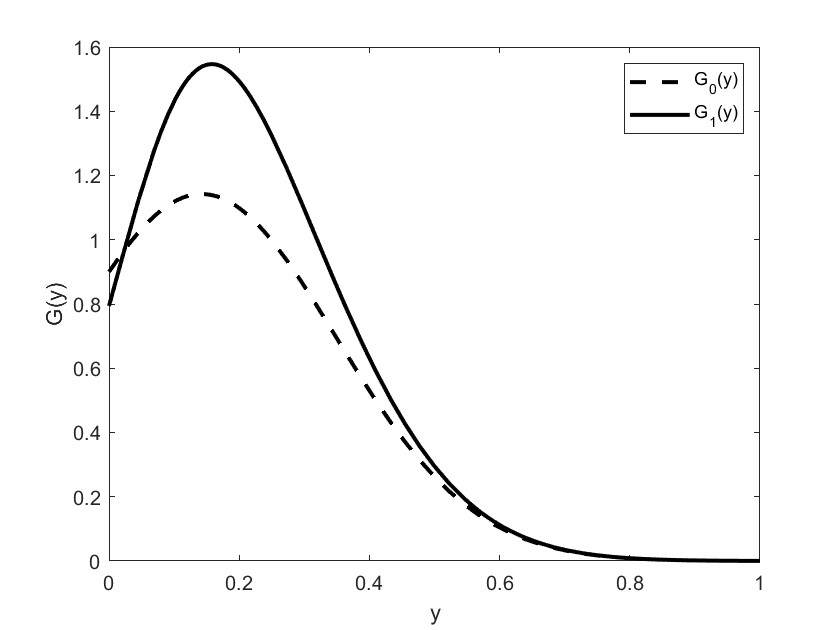}
  \caption{Survival density under the case $\mE[Z^0]<0$ (left panel), and under $\mE[Z^0]>0$ (right panel).\label{fig:survival_dominance}}
\end{figure}

The left panel corresponds to the case where condition \eqref{eq:condition1} holds, i.e., insurance is cheap and reinsurance is expensive. Here the survival probability of the strategy $\tilde{\mathbf{b}}$ dominates that of the strategy $\mathbf{b}$ for all values of $y$. In the right panel there exist a level $y^*>0$ (small) at which these two curves switch. However the area under the curve $G_1$ in the set $\{y>y^*\}$ largely compensates that in the set $\{y<y^*\}$. Such point $y^*$ corresponds to $y^*=\frac{2\mE[Z^0]}{1-2\rho}$ (and it only exists in case $\mE[Z^0]>0$). Note that, natural bounds on the value of $\eta$, i.e. $0 \le \eta \le \theta$, guarantee that such compensation of areas always applies and hence $p(\tilde{\mathbf{b}})>p(\mathbf{b})$, (see Proposition \ref{prop:optimal_strategy}).

\subsection{The penalisation problem}\label{sec:penalisation}
Suppose now, we have the following situation: the insurer may decide (at time zero) to have or to have not an update in the reinsurance contract at time $T/2$. It she updates the contract, she will pay a penalty amounting to $PT$ at time $T/2$. In case of no changes, no penalty will be applied. The strategies corresponding to these two different scenarios are chosen to achieve a Gaussian distribution at time $T$ with the same target variance $\delta^2T$. If the insurer does not change the strategy at time $T/2$ then the mean of the net collective is $M'T$, uniquely determined by the condition on the target variance. In case the strategy is changed at time $T/2$, the final expected wealth will be $M<M'$.
Next, we show that, due to the insurer's objective to minimise the ruin probability, changing the strategy at time $T/2$ is more preferable, even with a smaller expected mean.

We assume that $M=M'-P$.
Let $\hat{\mathbf{b}}=(\hat b, \hat b)$ be the strategy where the insurer decides to make no changes at time $T/2$ and let $\mathbf{b}=(b_0, b_1)$ and $\tilde{\mathbf{b}}=(b_1, b_0)$ be the admissible strategy where the insurer switches, with $b_0<b_1$. We know, by Proposition \ref{prop:optimal_strategy}, that strategy $\tilde{\mathbf{b}}$ is better than $\mathbf{b}$.
The survival probability of strategy $\hat{\mathbf{b}}$ is given by
\begin{align*}
p(\hat{\mathbf{b}})= \mP\left[\bar{X}^{\hat b}_{T/2}>0, \ \bar{X}^{\hat b}_T>0\right]
\end{align*}

We let $\hat Y=\bar{X}^{\hat b}_T$. Then we get that $Y\sim N(M'T, \delta^2T)$, and we observe that
\begin{align*}
\bar{X}^{\hat b}_{T/2}= \frac{1}{2} \hat Y + \hat Z,
\end{align*}
where
\begin{align*}
\hat Z\sim N\left(\lambda \mu (\theta \hat b-\theta+\eta)\frac{T}{2}-M'\frac{T}{2}, \frac{1}{4}\delta^2 T \right).
\end{align*}
Random variables $\hat Y$ and $\hat Z$ are independent. This implies that
\begin{align*}
p(\hat{\mathbf{b}})= \left[2\hat Z>-\hat Y, \hat Y>0\right]
=\int_{0}^{\infty}\left(1-\phi\left(\frac{-y-2\mE[\hat Z]}{\sqrt{\delta^2T}}\right)\right) f_{\hat Y}(y) \md y.
\end{align*}

Next, for the strategy $\tilde{\mathbf{b}}$ the survival probability is given by:
\begin{align*}
p(\tilde{\mathbf{b}})&= \int_{PT}^{\infty} \left(1-\phi\left(\frac{-y+PT-\frac{\mE\left[Z^1\right]}{\gamma}}{\sqrt{\delta^2 \frac{\rho}{\gamma}T}}\right)\right) f_{\hat Y}(y) \md  y,
\end{align*}
where $Z^1=N\left(\lambda\mu(\theta b_1-\theta+\eta)T/2-2\gamma (M'-P) T/2, \gamma(1-\gamma)\delta^2T\right)$, like in the proof of Proposition \ref{prop:optimal_strategy}.

It is clear that for $P=0$, there is a unique strategy, $\hat{\mathbf{{b}}}$, that allows to achieve the desired distribution for the net collective. For $P>0$, however the strategy $\hat{\mathbf{b}}$ has a survival probability that is always smaller than the survival probability of the optimal strategy $\mathbf{b}$ and larger than that of $\mathbf{\tilde b}$. This difference is illustrated in Figure \ref{fig:penalisation}.

\begin{figure}[ht]
  \centering
  \includegraphics[width=.6\textwidth]{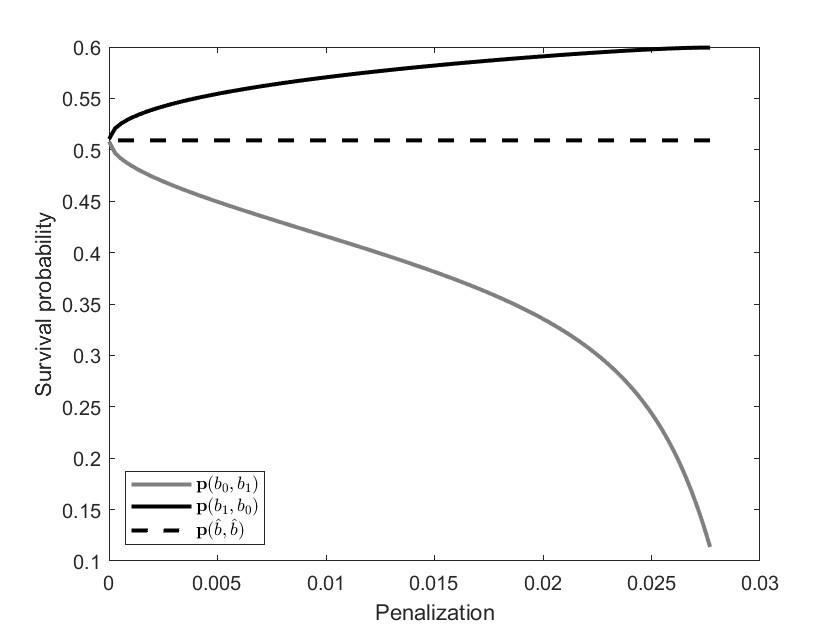}
  \caption{Survival probabilities under penalisation. Black line corresponds to the survival probability of the strategy $\tilde{\mathbf{b}}$, grey line to that of  $\mathbf{b}$ and the dashed line to that of the constant strategy $\hat{\mathbf{b}}$.}\label{fig:penalisation}
\end{figure}


\subsection{A 3 period model}\label{sec:3periods}
To explain the complexity of the problem for $n>2$, we consider the case $n=3$. Here, the form of the survival probability does not allow to derive conditions that ensure a clear dominance of one strategy. In addition, the computational time increases with the number of periods.

To give some intuition on how to deal with this case, we restrict to deterministic strategies $\mathbf{b}=(b_0, b_1, b_2)$. Then, it holds that
\begin{align*}
 b_0 + b_1 + b_2&=3\frac{M+\lambda \mu (\theta-\eta)}{\lambda \mu \theta}\;,\\
 b_0^2 + b_1^2 + b_2^2&=\frac{3\delta^2}{\lambda \mu_2}\;,\\
b_0, b_1, b_2& \in [0,1]
\end{align*}
which means that there are infinitely many combinations of $(b_0, b_1, b_2)$ that lead to the target distribution.

In particular, admissible triplets build (a part of) a circle as shown in Figure \ref{fig:circle}.

\begin{figure}[ht]
  \centering
  \includegraphics[width=.6\textwidth]{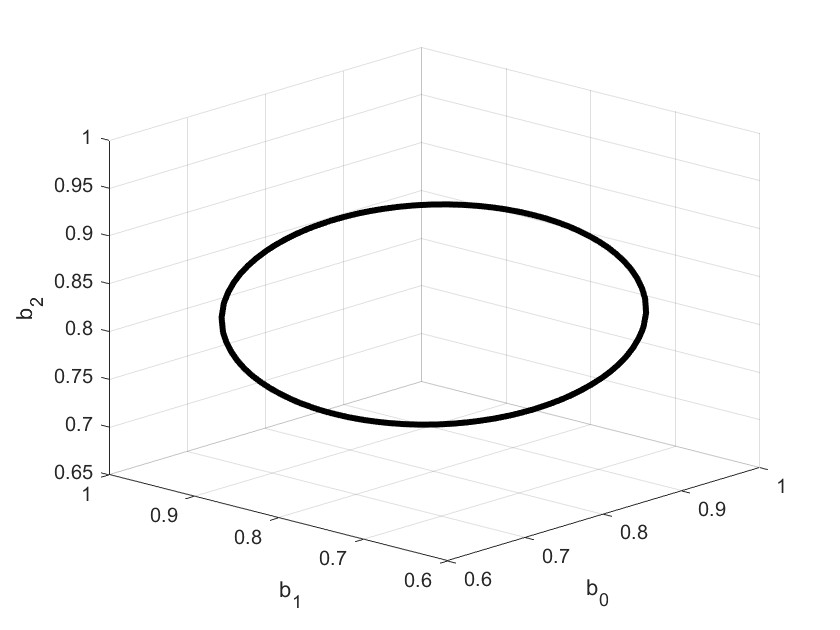}
  \caption{Admissible deterministic strategies for $n=3$ and parameters $T=1;
\mu=0.15;
\mu_2=0.06;
\lambda=1;
\theta=0.35;
\eta=0.2;
M=0.02;
\delta=0.2;$.}\label{fig:circle}
\end{figure}

To choose the ruin-minimising strategy we look at survival probability

\begin{align*}
p(\mathbf{b})=\mP\left[\bar X^{b_0}_{\frac{T}{3}}>0, \bar X^{b_0, b_1}_{\frac{2T}{3}}>0, \bar X^{\mathbf{b}}_{T}>0 \right]\;.
\end{align*}

We define auxiliary random variables $\zeta^0, \zeta^1$ such that
\begin{align*}
&\zeta^0\sim N\left(\lambda \mu (\theta b_0-\theta+\eta)\frac{T}{3}-M\rho_0T, \ \rho_0(1-\rho_0) \delta^2 T \right)\;,\\
&\rho_2\zeta^0+\zeta^1 \sim  N\left(-\lambda \mu \theta b_2\frac{T}{3}+\lambda \mu (\theta-\eta)T+(1-\rho_1)MT, \ \rho_1 (1-\rho_1) \delta^2 T\right)\;,
\end{align*}
which are correlated. Then, we have that

\begin{align*}
p(\mathbf{b})&=\mP\left[\frac{\zeta^0}{\rho_0}>- \bar X^{\mathbf{b}}_T, \ \frac{\rho_2\zeta^0+\zeta^1}{\rho_1} >-\bar X^{\mathbf{b}}_T , \ \bar X^{\mathbf{b}}_{T}>0 \right]
\\&=\int_0^{\infty} \mP\left[\frac{\zeta^0}{\rho_0}>y, \ \frac{\rho_2\zeta^0+\zeta^1}{\rho_1} >y\right] f_{Y}(y)\md y\;.
\end{align*}
where $Y\sim N(MT, \delta^2 T)$, and $f_Y(y)$ is the corresponding density. 

\begin{figure}[ht]
  \centering
  \includegraphics[width=.6\textwidth]{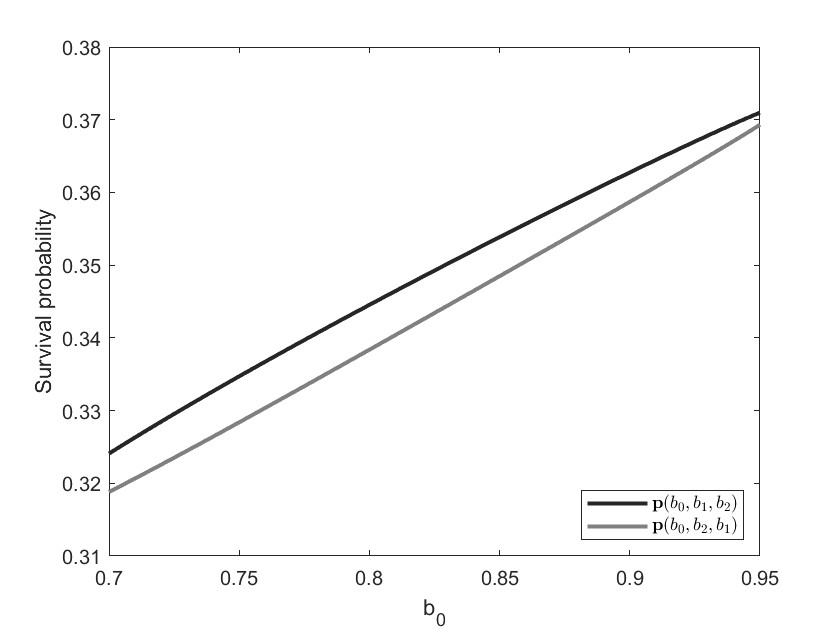}
  \caption{Survival probabilities as functions of the first component $b_0$ of the strategy $(b_0, b_1, b_2)$. \label{fig:n3}}
\end{figure}
Figure \ref{fig:n3} shows the survival probability with respect to the first component $b_0$. It is clear that, once $b_0$ is chosen, there are only two possible choices for $b_1$ and $b_2$. Suppose that, for instance $b_1>b_2$, then for a fixed $b_0$, the possible strategies are $(b_0, b_1, b_2)$ and $(b_0, b_1, b_2)$. We see that the survival probability is maximised by the largest $b_0$ and the combination that leads to the higher survival probability is the sorted one, i.e.\ $(b_0, b_1, b_2)$ with $b_0>b_1>b_2$.

We conclude the section by showing that the sorted sequence $\mathbf{b}=(b_0, b_1, b_2)$ leads to a bigger survival probability than the ``unsorted'' sequence $\tilde{\mathbf{b}}=(b_0, b_2, b_1)$. This means, in particular, that like shown in Figure \ref{fig:n3} any unsorted sequence will be overperformed by a sorted one.
\\To prove this, we denote by $p^x(\cdot)$ the survival probability of a strategy where the initial capital is $x$.
Then,
\begin{align*}
p(\mathbf{b})&=\mP\left[\bar X^{b_0}_{T/3}>0,  \bar X^{(b_0, b_1)}_{2/3T}>0, \bar X_T^{(b_0, b_1, b_2)}>0\right]=\mE\left[\one_{\bar X^{b_0}_{T/3}>0}  \ p^{\bar X^{b_0}_{T/3}}((b_1, b_2))\right]\\
&>\mE\left[\one_{\bar X^{b_0}_{T/3}>0} \  p^{\bar X^{b_0}_{T/3}}((b_2, b_1))\right]=p(\tilde{\mathbf{b}})\;,
\end{align*}
where the inequality follows from the case $n=2$.

\section{Conclusions}\label{sec:conclusion}
In this paper, we consider an insurance company whose objective is to choose a dividend payment or a reinsurance strategy leading to a certain surplus distribution. The question which strategy to prefer depends on the underlying target functional -- the value of expected discounted dividends (to be maximised) or the ruin probability (to be minimised).
Such a problem is motivated by the necessity of being able to compute risk measures, typically based on the distribution of a future loss at some fixed date. Fixing a terminal wealth distribution would allow to compute several risk measures at once, instead of choosing a specific constraint in the beginning of an optimisation task. In this line, the present paper represents the first step towards a more detailed and more realistic analysis of the problems faced by practitioners on the almost daily basis.
\\In addition, we would like to stress that the dividend related problems can be easily generalised to a continuous time framework. In the reinsurance setting we are able to fully analyse the 2-period problem. Since, reinsurance contracts are usually difficult or even impossible to update before the maturity date, this setting seems to be the most realistic one from a practical point of view.

Using a pool of possible distributions and mean/variance combinations instead of only one specific distribution is a possible extension direction. However, in this paper our main target is to introduce an idea and to illustrate with 2 clear settings how this idea can be implemented. For instance, in the reinsurance setting, the discrete nature of the problem does not allow to use the differential equation approach. Any return function would depend on the initial surplus and on the time. Changing the length of one interval would completely change the optimal strategy, as more weight will be put on the remaining intervals. Therefore, we are using purely probabilistic methods to prove our claims for the 2-period case. In the general n-period setting, the admissible strategies are not necessarily deterministic, and the optimal strategy may even not exist.

Our future research will concentrate on the extensions of the presented models. We plan to work on the n-period model for the reinsurance setting. We will also consider problems with non-normal target distributions and allow continuous time ruin-checks.

\section*{Acknowledgements}
The work of K. Colaneri and B. Salterini has been partially supported by Indam-Gnampa though the project U-UFMBAZ-2020-000791. Part of this work has been done while K. Colaneri and B. Salterini were visiting TU Vienna.
\medskip
\\\indent The research of Julia Eisenberg was funded by the Austrian Science Fund (FWF), Project number V 603-N35.

\bibliographystyle{plainnat}

\end{document}